%% file: optimistic-approachability.tex
\def\year{2021}\relax
\documentclass[letterpaper]{article} 
\usepackage{aaai21}  
\usepackage{times}  
\usepackage{helvet} 
\usepackage{courier}  
\usepackage[hyphens]{url}  
\usepackage{graphicx} 
\usepackage{natbib}  
\usepackage{caption} 
\usepackage[utf8]{inputenc}
\usepackage[T1]{fontenc}

\title{Faster Game Solving via Predictive Blackwell Approachability: Connecting Regret Matching and Mirror Descent}
\author{
Gabriele Farina,\textsuperscript{\rm 1}
Christian Kroer,\textsuperscript{\rm 2}
Tuomas Sandholm\textsuperscript{\rm 1,\rm 3,\rm 4,\rm 5}\\
}
\affiliations {
\textsuperscript{\rm 1}Computer Science Department, Carnegie Mellon University,
\textsuperscript{\rm 2}IEOR Department, Columbia University\\
\textsuperscript{\rm 3}Strategic Machine, Inc.,
\textsuperscript{\rm 4}Strategy Robot, Inc.,
\textsuperscript{\rm 5}Optimized Markets, Inc.\\
gfarina@cs.cmu.edu, christian.kroer@columbia.edu, sandholm@cs.cmu.edu
}
\date{\today}

\usepackage{amsmath}
\usepackage{float}
\usepackage{amssymb}
\usepackage{amsthm}
\usepackage{mathtools}
\usepackage{xcolor}
\usepackage{bm}
\usepackage{etoolbox}
\usepackage[hidelinks,colorlinks=false]{hyperref}
\usepackage[capitalise,noabbrev]{cleveref}

\usepackage{tcolorbox}
\usepackage{letltxmacro}
\usepackage{mleftright}
\usepackage{tikz}
\usepackage[outline]{contour}
\usepackage{enumitem}
\usepackage{nicefrac}
\usepackage{comment}
\usepackage{booktabs}
\usepackage{thm-restate}
\usepackage{environ}
\usepackage{graphbox}
\usepackage{siunitx}
\allowdisplaybreaks

\newcommand{\cX}{\mathcal{X}}
\newcommand{\cY}{\mathcal{Y}}
\newcommand{\cL}{\mathcal{L}}
\newcommand{\cJ}{\mathcal{J}}
\newcommand{\cC}{\mathcal{C}}
\newcommand{\cR}{\mathcal{R}}
\newcommand{\cK}{\mathcal{K}}
\newcommand{\cD}{\mathcal{D}}
\newcommand{\bbR}{\mathbb{R}}
\newcommand{\bbB}{\mathbb{B}}
\newcommand{\defeq}{\coloneqq}
\renewcommand{\vec}[1]{\bm{#1}}
\newcommand{\mat}[1]{\bm{#1}}
\DeclareMathOperator*{\argmin}{arg\,min}

\usepackage[linesnumbered,ruled,lined,noend]{algorithm2e}
\makeatletter
    \patchcmd\algocf@Vline{\vrule}{\vrule \kern-0.4pt}{}{}
    \patchcmd\algocf@Vsline{\vrule}{\vrule \kern-0.4pt}{}{}
\makeatother
\SetKwComment{Hline}{}{\vspace{-3mm}\textcolor{gray}{\hrule}\vspace{1mm}}
\definecolor{darkgrey}{gray}{0.3}
\definecolor{commentcolor}{gray}{0.3}
\SetKwComment{Comment}{\color{commentcolor}$\triangleright$\ }{}
\SetCommentSty{}
\SetNlSty{}{\color{darkgrey}}{}
\crefname{algocf}{Algorithm}{Algorithms}
\SetKwProg{Fn}{function}{}{}
\SetKwProg{Subr}{subroutine}{}{}

\crefalias{AlgoLine}{line}%

\makeatletter
\let\cref@old@stepcounter\stepcounter
\def\stepcounter#1{%
  \cref@old@stepcounter{#1}%
  \cref@constructprefix{#1}{\cref@result}%
  \@ifundefined{cref@#1@alias}%
    {\def\@tempa{#1}}%
    {\def\@tempa{\csname cref@#1@alias\endcsname}}%
  \protected@edef\cref@currentlabel{%
    [\@tempa][\arabic{#1}][\cref@result]%
    \csname p@#1\endcsname\csname the#1\endcsname}}
\makeatother

\newtheorem{corollary}{Corollary}[]
\newtheorem{definition}{Definition}[]

\newtheorem{lemma}{Lemma}[]

\newtheorem{theorem}{Theorem}[]

\usetikzlibrary{calc,patterns,arrows.meta}
\newcommand*\circled[1]{\tikz[baseline=(char.base)]{
            \node[shape=circle,draw,inner sep=1pt] (char) {\scriptsize #1};}}
\newcommand\numberthis{\addtocounter{equation}{1}\tag{\theequation}}
\makeatletter
\newcounter{game}
\crefname{game}{Game}{Games}

\newtoggle{create@game@label}
\newcommand{\@@addplots}[3]{
    \node at (0,-2.41) {\scalebox{.75}[.73]{\includegraphics{plots/#2_error_#1.pdf}}};
    \node at (0, 0) {\scalebox{.75}[.73]{\includegraphics{plots/#2_regret_#1.pdf}}};
    \node (title) at (.35, 1.31) {\small \ref*{game:#2}~~#3};
}
\newcommand{\addplot@y}[3]{
    \iftoggle{create@game@label}{\label{game:#2}}{}%
    \begin{minipage}[t]{4.7cm}%
        \begin{tikzpicture}[baseline=(title.base)]
            \@@addplots{#1}{#2}{#3}%
            \node[rotate=90] at (-2.35, -.1) {\fontsize{8}{8}\selectfont Nash gap};
\node[rotate=90] at (-2.35, -2.45) {\fontsize{8}{8}\selectfont Avg. $\|\vec{\ell}^t - \vec{m}^t\|_2$};
        \end{tikzpicture}%
    \end{minipage}%
}
\newcommand{\addplot@noy}[3]{%
    \iftoggle{create@game@label}{\label{game:#2}}{}%
    \begin{minipage}[t]{4.4cm}%
        \begin{tikzpicture}[baseline=(title.base)]%
            \@@addplots{#1}{#2}{#3}
        \end{tikzpicture}%
    \end{minipage}%
}

\newcommand{\addplotA}{\toggletrue{create@game@label}\refstepcounter{game}%
\@ifstar{\addplot@noy{pcfrp_vs_sota}}{\addplot@y{pcfrp_vs_sota}}}
\newcommand{\addplotB}{\togglefalse{create@game@label}\@ifstar{\addplot@noy{lin_vs_quad_avg}}{\addplot@y{lin_vs_quad_avg}}}
\newcommand{\addplotC}{\togglefalse{create@game@label}\@ifstar{\addplot@noy{pdfcr}}{\addplot@y{pdfcr}}}
\makeatother

\newcommand{\linesty}[1]{\raisebox{1mm}{\tikz\draw[very thick,#1] (0,0)--(.55,0);}}
\definecolor{pred}{HTML}{FF0000}
\definecolor{pblue}{HTML}{0000FF}
\definecolor{pgreen}{HTML}{008000}
\definecolor{ppurple}{HTML}{800080}
\definecolor{pviolet}{HTML}{EE82EE}
\newcommand{\makelegendA}{\small
\fbox{
    Legend:\qquad
     \linesty{pred,dash dot}~\pcfrp{}\quad
     \linesty{pblue}~\cfrp{}\quad
     \linesty{pgreen,dashed}~LCFR\quad
     \linesty{ppurple,dotted}~DCFR
}
}
\newcommand{\makelegendB}{\small
\fbox{
    Legend:\qquad
     \begin{tabular}{ll}
     \linesty{pgreen,dotted}~\pcfrp{} (linear averaging)&
     \linesty{pred,dash dot}~\pcfrp{} (quadratic averaging)\\
     \linesty{pblue}~\cfrp{} (linear averaging)&
     \linesty{pviolet,dashed}~\cfrp{} (quadratic averaging)
     \end{tabular}
}
}
\newcommand{\makelegendC}{\small
\fbox{
    Legend:\qquad
     \begin{tabular}{ll}
     \linesty{pblue}~\cfrp{}&
     \linesty{ppurple,dotted}~DCFR\\
     \linesty{pred,dash dot}~\pcfrp{}&
     \linesty{pred,dashed}~\pcfrp{} (Quad. avg. prediction)\\
     \linesty{black}~Predictive DFCR&
     \linesty{black,dashed}~Predictive DCFR (Quad. avg. prediction)\\
     \end{tabular}
}
}

\LetLtxMacro{\baseproof}{\proof}
\LetLtxMacro{\endbaseproof}{\endproof}

\tcbuselibrary{skins,breakable}
\NewEnviron{shadedproof}{%
    \begin{tcolorbox}[enhanced jigsaw,breakable,colback=gray!15,sharp corners,frame hidden]
        \begin{baseproof}
            \BODY
        \end{baseproof}
    \end{tcolorbox}
}

\newcommand{\xhat}{\hat{\vec{x}}}
\newcommand{\zhat}{\hat{\vec{z}}}
\newcommand{\thetahat}{\hat{\vec{\theta}}}
\newcommand{\regu}{\varphi}
\renewcommand{\div}[2]{D_\regu(#1 \,\|\, #2)}

\newcommand{\rmp}{RM$^+$}
\newcommand{\prmp}{PRM$^+$}
\newcommand{\cfrp}{CFR$^+$}
\newcommand{\pcfrp}{PCFR$^+$}

\usepackage[switch]{lineno}
\makeatletter
    \AtBeginDocument{%
      \@ifpackageloaded{amsmath}{%
        \newcommand*\patchAmsMathEnvironmentForLineno[1]{%
          \expandafter\let\csname old#1\expandafter\endcsname\csname #1\endcsname
          \expandafter\let\csname oldend#1\expandafter\endcsname\csname end#1\endcsname
          \renewenvironment{#1}%
                           {\linenomath\csname old#1\endcsname}%
                           {\csname oldend#1\endcsname\endlinenomath}%
        }%
        \newcommand*\patchBothAmsMathEnvironmentsForLineno[1]{%
          \patchAmsMathEnvironmentForLineno{#1}%
          \patchAmsMathEnvironmentForLineno{#1*}%
        }%
        \patchBothAmsMathEnvironmentsForLineno{equation}%
        \patchBothAmsMathEnvironmentsForLineno{align}%
        \patchBothAmsMathEnvironmentsForLineno{flalign}%
        \patchBothAmsMathEnvironmentsForLineno{alignat}%
        \patchBothAmsMathEnvironmentsForLineno{gather}%
        \patchBothAmsMathEnvironmentsForLineno{multline}%
      }{}
    }
\makeatother
\begin{document}
    \maketitle

    \begin{abstract}
      Blackwell approachability is a framework for reasoning about repeated games with vector-valued payoffs.
      We introduce \emph{predictive} Blackwell approachability, where an estimate of the next payoff vector is given, and the decision maker tries to achieve better performance based on the accuracy of that estimator.
       In order to derive algorithms that achieve predictive Blackwell approachability, we start by showing a powerful connection between four well-known algorithms. \emph{Follow-the-regularized-leader (FTRL)} and \emph{online mirror descent (OMD)} are the most prevalent regret minimizers in online convex optimization. In spite of this prevalence, the \emph{regret matching (RM)} and \emph{regret matching$^+$ (\rmp)} algorithms have been preferred in the practice of solving large-scale games (as the local regret minimizers within the counterfactual regret minimization framework).
       We show that RM and \rmp{} are the algorithms that result from running FTRL and OMD, respectively, to select the halfspace to force at all times in the underlying Blackwell approachability game.
       By applying the predictive variants of FTRL or OMD to this connection, we obtain predictive Blackwell approachability algorithms, as well as predictive variants of RM and \rmp.
      In experiments across 18 common zero-sum extensive-form benchmark games, we show that predictive \rmp{} coupled with counterfactual regret minimization converges vastly faster than the fastest prior algorithms (\cfrp{}, DCFR, LCFR) across all games but two of the poker games, sometimes by two or more orders of magnitude.
    \end{abstract}

    \input{text/introduction}
    \input{text/regret_minimizers}

    \input{text/blackwell_approachability}
    \input{text/rm_to_approachability}
    \input{text/rm_and_rmplus}

    \input{text/optimistic_approachability}

    \input{text/experiments}
    \input{text/conclusions}
    \section*{Acknowledgments}
    This material is based on work supported by the National Science Foundation under grants IIS-1718457, IIS-1901403, and CCF-1733556, and the ARO under award W911NF2010081. Gabriele Farina is supported by a Facebook fellowship.

    \bibliographystyle{aaai21}
    \bibliography{dairefs}

\iftrue
    \clearpage
    \onecolumn
    \leftlinenumbers
    \appendix

    \makeatletter
        \renewcommand{\section}{%
          \@startsection{section}{1}{\z@}%
                        {-2.0ex \@plus -0.5ex \@minus -0.2ex}%
                        { 1.5ex \@plus  0.3ex \@minus  0.2ex}%
                        {\Large\bf\raggedright}%
        }
        \renewcommand{\subsection}{%
          \@startsection{subsection}{2}{\z@}%
                        {-1.8ex \@plus -0.5ex \@minus -0.2ex}%
                        { 0.8ex \@plus  0.2ex}%
                        {\large\bf\raggedright}%
        }
    \makeatother

        \section*{Additional Bibliographic Remarks}
        \begin{enumerate}
            \item Gordon's Lagrangian Hedging framework \citep{Gordon05:NoRegret,Gordon07:NoRegret}
        partially overlaps with the construction by \citet{Abernethy11:Blackwell}
        that we used in the paper. It appears that Abernethy et al. were not
        aware of Gordon's results. We did not investigate to what extent the \emph{predictive}
        point of view we adopted in the paper could apply to Gordon's result.

            \item In his PhD thesis, \citet{Burch18:Time} mentions an algorithm that
                he coins ``optimistic \rmp{}''. No theory is provided, and unfortunately Burch never
                defined the algorithm formally, so it is not clear whether his algorithm is the same as \prmp{} as defined in
                \cref{algo:prmp} in our paper. \citet{Brown17:Dynamic} gave an interpretation of
                optimistic \rmp{} by Burch that would imply it is different from \prmp{}. We indend to check with Burch directly for the final version of this paper.
        \end{enumerate}

        \input{text/appendix_ftrl_omd}
        \input{text/appendix_olo_to_approachability}
        \input{text/appendix_rm_rmplus}

        \input{text/appendix_predictive_approachability}
        \input{text/appendix_cfr}
        \input{text/appendix_experiments}

\fi
\end{document}

%% file: text/introduction.tex
\section{Introduction}

Extensive-form games (EFGs) are the standard class of games that can be used to model sequential interaction, outcome uncertainty, and imperfect information.
Operationalizing these models requires algorithms for computing game-theoretic equilibria. A recent success of EFGs is the use of Nash equilibrium for several recent poker AI milestones, such as essentially solving the game of limit Texas hold'em~\citep{Bowling15:Heads}, and beating top human poker pros in no-limit Texas hold'em with the \emph{Libratus} AI~\citep{Brown17:Superhuman}.
%
%
%
A central component of all recent poker AIs has been a fast iterative method for computing approximate Nash equilibrium at scale.
The leading approach is the \emph{counterfactual regret minimization (CFR)} framework, where the problem of minimizing regret over a player's strategy space of an EFG is decomposed into a set of regret-minimization problems over probability simplexes~\citep{Zinkevich07:Regret,Farina19:Regret}. Each simplex represents the probability over actions at a given decision point. The CFR setup can be combined with any regret minimizer for the simplexes. If both players in a zero-sum EFG repeatedly play each other using a CFR algorithm, the average strategies converge to a Nash equilibrium. Initially \emph{regret matching} (RM) was the prevalent simplex regret minimizer used in CFR. Later, it was found that by alternating strategy updates between the players, taking linear averages of strategy iterates over time, and using a variation of RM called \emph{regret-matching$^+$ (\rmp)}~\citep{Tammelin14:Solving} leads to significantly faster convergence in practice. This variation is called \cfrp{}. Both CFR and \cfrp{} guarantee convergence to Nash equilibrium at a rate of $T^{-1/2}$.
\cfrp{} has been used in every milestone in developing poker AIs in the last decade~\citep{Bowling15:Heads,Moravvcik17:DeepStack,Brown17:Superhuman,Brown19:Superhuman}. This is in spite of the fact that its theoretical rate of convergence is the same as that of CFR with RM~\citep{Tammelin14:Solving,Farina19:Online,Burch19:Revisiting}, and there exist algorithms which converge at a faster rate of $T^{-1}$~\citep{Hoda10:Smoothing,Kroer20:Faster,Farina19:Optimistic}.
In spite of this theoretically-inferior convergence rate, \cfrp{} has repeatedly performed favorably against $T^{-1}$ methods for EFGs~\cite{Kroer18:Solving,Kroer20:Faster,Farina19:Optimistic,Gao19:Increasing}. Similarly, the \emph{follow-the-regularized-leader (FTRL)} and \emph{online mirror descent (OMD)} regret minimizers, the two most prominent algorithms in online convex optimization, can be instantiated to have a better dependence on dimensionality than \rmp{} and RM, yet \rmp{} has been found to be superior~\citep{Brown17:Dynamic}.

There has been some interest in connecting RM to the more prevalent (and more general) online convex optimization algorithms such as OMD and FTRL, as well as classical first-order methods.
\citet{Waugh15:Unified} showed that RM is equivalent to Nesterov's dual averaging algorithm (which is an offline version of FTRL), though this equivalence requires specialized step sizes that are proven correct by invoking the correctness of RM itself.
\citet{Burch18:Time} studies RM and \rmp{}, and contrasts them with mirror descent and other prox-based methods.

We show a strong connection between RM, \rmp{}, and FTRL, OMD. This connection arises via \emph{Blackwell approachability}, a framework for playing games with vector-valued payoffs, where the goal is to get the average payoff to approach some convex target set.
Blackwell originally showed that this can be achieved by repeatedly \emph{forcing} the payoffs to lie in a sequence of halfspaces containing the target set~\cite{Blackwell56:analog}.
Our results are based on extending an equivalence between approachability and regret minimization~\cite{Abernethy11:Blackwell}.
We show that RM and \rmp{} are the algorithms that result from running FTRL and OMD, respectively, to select the halfspace to force at all times in the underlying Blackwell approachability game. The equivalence holds for any constant step size.
Thus, RM and \rmp{}, the two premier regret minimizers in EFG solving, turn out to follow exactly from the two most prevalent regret minimizers from online optimization theory.
This is surprising for several reasons:
\begin{itemize}[leftmargin=5mm,itemsep=1mm]
\item
  \rmp{} was originally discovered as a heuristic modification of RM in order to avoid accumulating large negative regrets.
   In contrast, OMD and FTRL were developed separately from each other.
 \item
   When applying FTRL and OMD directly to the strategy space of each player, \citet{Farina19:Optimistic,Farina20:Stochastic} found that FTRL seems to perform better than OMD, even when using stochastic gradients. This relationship is reversed here, as \rmp{} is \emph{vastly} faster numerically than RM.
 \item
   The dual averaging algorithm (whose simplest variant is an offline version of FTRL), was originally developed in order to have increasing weight put on more recent gradients, as opposed to OMD which has constant or decreasing weight~\citep{Nesterov09:Primal}. Here this relationship is reversed: OMD (which we show has a close link to \rmp{}) thresholds away old negative regrets, whereas FTRL keeps them around. Thus OMD ends up being \emph{more} reactive to recent gradients in our setting.
 \item
   FTRL and OMD both have a step-size parameter that needs to be set according to the magnitude of gradients, while RM and \rmp{} are parameter free (which is a desirable feature from a practical perspective).
   To reconcile this seeming contradiction, we show that the step-size parameter does not affect which halfspaces are forced, so any choice of step size leads to RM and \rmp{}.
\end{itemize}

Leveraging our connection, we study the algorithms that result from applying predictive variants of FTRL and OMD to choosing which halfspace to force. By applying predictive OMD we get the first predictive variant of \rmp{}, that is, one that has regret that depends on how good the sequence of predicted regret vectors is (as a side note of their paper, \citet{Brown19:Solving} also tried a heuristic for optimism/predictiveness by counting the last regret vector twice in \rmp{}, but this does not yield a predictive algorithm). We call our regret minimizer \emph{predictive regret matching$^+$} (\prmp{}).
We go on to instantiate CFR with \prmp{} using the two standard techniques---alternation and quadratic averaging----and find that it often converges much faster than \cfrp{} and every other prior CFR variant, sometimes by several orders of magnitude.
We show this on a large suite of common benchmark EFGs. However, we find that on poker games (except shallow ones), \emph{discounted CFR (DCFR)}~\citep{Brown19:Solving} is the fastest. We conclude that our algorithm based on \prmp{} yields the new state-of-the-art convergence rate for the remaining games. Our results also highlight the need to test on EFGs other than poker, as our non-poker results invert the superiority of prior algorithms as compared to recent results on poker.


%% file: text/regret_minimizers.tex
\section{Online Linear Optimization,\\Regret Minimizers, and Predictions}
\input{text/ftrl_omd_algos}

At each time $t$, an oracle for the \emph{online linear optimization (OLO)} problem supports the following two operations, in order:
\textsc{NextStrategy} returns a point $\vec{x}^t \in \cD \subseteq \bbR^n$, and
\textsc{ObserveLoss} receives a \emph{loss vector} $\vec{\ell}^t$ that is meant to evaluate the strategy $\vec{x}^t$ that was last output. Specifically, the oracle incurs a loss equal to $\langle \vec{\ell}^t, \vec{x}^t\rangle$. The loss vector $\vec{\ell}^t$ can depend on all past strategies that were output by the oracle.
The oracle operates \emph{online} in the sense that each strategy $\vec{x}^t$ can depend only on the decision $\vec{x}^1,\dots,\vec{x}^{t-1}$ output in the past, as well as the loss vectors $\vec{\ell}^1, \dots,\vec{\ell}^{t-1}$ that were observed in the past. No information about the future losses $\vec{\ell}^t,\vec{\ell}^{t+1},\dots$ is available to the oracle at time $t$.
The objective of the oracle is to make sure the \emph{regret}
\[
    R^T(\hat{\vec{x}}) \defeq \sum_{t=1}^T \langle\vec{\ell}^t, \vec{x}^t\rangle - \sum_{t=1}^T \langle \vec{\ell}^t, \hat{\vec{x}}\rangle =
    \sum_{t=1}^T \langle \vec{\ell}^t, \vec{x}^t - \hat{\vec{x}}\rangle,
\]
which measures the difference between the total loss incurred up to time $T$ compared to always using the \emph{fixed} strategy $\hat{\vec{x}}$,
does not grow too fast as a function of time $T$. Oracles that guarantee that $R^T(\hat{\vec{x}})$ grow sublinearly in $T$ in the worst case for all $\hat{\vec{x}} \in \cD$ (no matter the sequence of losses $\vec{\ell}^1, \dots, \vec{\ell}^T$ observed) are called \emph{regret minimizers}.
While most theory about regret minimizers is developed under the assumption that the domain $\cD$ is \emph{convex} and \emph{compact}, 
in this paper we will need to consider sets $\cD$ that are convex and closed, but unbounded (hence, not compact). 

\subsection{Incorporating Predictions} A recent trend in online learning has been concerned with constructing oracles that can incorporate \emph{predictions} of the next loss vector $\vec{\ell}^t$ in the decision making~\citep{Chiang12:Online,Rakhlin13:Online,Rakhlin13:Optimization}. Specifically, a  \emph{predictive} oracle differs from a regular (that is, non-predictive) oracle for OLO in that the \textsc{NextStrategy} function receives a \emph{prediction} $\vec{m}^t \in \bbR^n$ of the next loss $\vec{\ell}^t$ at all times $t$.
Conceptually, a ``good'' predictive regret minimizer should guarantee a superior regret bound than a non-predictive regret minimizer if $\vec{m}^t \approx \vec{\ell}^t$ at all times $t$. Algorithms exist that can guarantee this. For instance, it is always possible to construct an oracle that guarantees that $R^T = O(1 + \sum_{t=1}^T \|\vec{\ell}^t - \vec{m}^t\|^2)$, which implies that the regret stays constant when $\vec{m}^t$ is clairvoyant. In fact, even stronger regret bounds can be attained: for example, \citet{Syrgkanis15:Fast} show that the sharper \emph{Regret bounded by Variation in Utilities (RVU)} condition can be attained, while \citet{Farina19:Stable} focus on \emph{stable-predictivity}.


\subsection{FTRL, OMD, and their Predictive Variants}
\emph{Follow-the-regularized-leader (FTRL)}~\citep{Schwartz07:Primal} and \emph{online mirror descent (OMD)} are the two best known oracles for the online linear optimization problem. Their \emph{predictive} variants are relatively new and can be traced back to the works by \citet{Rakhlin13:Online} and \citet{Syrgkanis15:Fast}. Since the original FTRL and OMD algorithms correspond to predictive FTRL and predictive OMD when the prediction $\vec{m}^t$ is set to the $\vec{0}$ vector at all $t$, the implementation of FTRL in \cref{algo:predictive ftrl} and OMD in \cref{algo:predictive omd} captures both algorithms. In both algorithm, $\eta > 0$ is an arbitrary step size parameter, $\cD \subseteq \bbR^n$ is a convex and closed set, and $\regu : \cD \to \bbR_{\ge 0}$ is a $1$-strongly convex differentiable regularizer (with respect to some norm $\|\cdot\|$). The symbol $\div{}{}$ used in OMD denotes the \emph{Bregman divergence} associated with $\regu$, defined as $\div{\vec{x}}{\vec{c}} \defeq \regu(\vec{x}) - \regu(\vec{c}) - \langle \nabla\regu(\vec{c}), \vec{x} - \vec{c}\rangle$ for all $\vec{x},\vec{c}\in \cD$.

We state regret guarantees for (predictive) FTRL and (predictive) OMD in \cref{prop:oco bound}. Our statements are slightly more general than those by \citet{Syrgkanis15:Fast}, in that we (i) do not assume that the domain is a simplex, and (ii) do not use quantities that might be unbounded in non-compact domains $\cD$. A proof of the regret bounds is in \cref{app:proofs ftrl} of the full version of the paper\footnote{The full version of this paper is at \url{arxiv.org/abs/2007.14358}.} for FTRL and \cref{app:proofs omd} for OMD.

\begin{restatable}{proposition}{propocobounds}\label{prop:oco bound}
    At all times $T$, the regret cumulated by (predictive) FTRL (\cref{algo:predictive ftrl}) and (predictive) OMD (\cref{algo:predictive omd}) compared to \emph{any} strategy $\hat{\vec{x}} \in \cD$ is bounded as
    \begin{equation*}
        R^T(\xhat) \le \frac{\regu(\xhat)}{\eta} + \eta\sum_{t=1}^T \|\vec{\ell}^t - \vec{m}^t\|_*^2 - \frac{1}{c\eta} \sum_{t=1}^{T-1} \|\vec{x}^{t+1} - \vec{x}^t\|^2,
    \end{equation*}
    where $c = 4$ for FTRL and $c = 8$ for OMD, and where $\|\cdot\|_*$ denotes the dual of the norm $\|\cdot\|$ with respect to which $\regu$ is $1$-strongly convex.
\end{restatable}

    \cref{prop:oco bound} implies that, by appropriately setting the step size parameter (for example, $\eta = T^{-1/2}$), (predictive) FTRL and (predictive) OMD guarantee $R^T(\xhat) = O(T^{1/2})$ for all $\xhat$. Hence, (predictive) FTRL and (predictive) OMD are regret minimizers.

%% file: text/ftrl_omd_algos.tex
\begin{figure*}\centering
    \begin{minipage}[t]{.490\linewidth}\small
        \makeatletter\let\@latex@error\@gobble\makeatother
        \SetInd{0.25em}{0.4em}
        \begin{algorithm}[H]\caption{(Predictive) FTRL}\label{algo:predictive ftrl}
            \DontPrintSemicolon
            $\vec{L}^0 \gets \vec{0} \in \bbR^n$\;
            \Hline{}
            \Fn{\normalfont\textsc{NextStrategy}($\vec{m}^{t}$)}{
                \Comment{\color{commentcolor} Set $\vec{m}^t = \vec{0}$ for non-predictive version}\vspace{1mm}
                \Return{$\displaystyle\argmin_{\xhat \in \cD} \mleft\{\langle \vec{L}^{t-1} + \vec{m}^t, \xhat\rangle + \frac{1}{\eta}\regu(\xhat)\mright\}$}\hspace*{-4cm}\;\label{line:ftrl next strategy}\vspace{1mm}
            }
            \Hline{}
            \Fn{\normalfont\textsc{ObserveLoss}($\vec{\ell}^{t}$)}{
                \vspace{2mm}$\vec{L}^{t} \gets \vec{L}^{t-1} + \vec{\ell}^t$\vspace{2mm}\;
            }
        \end{algorithm}
    \end{minipage}
    \hfill
    \begin{minipage}[t]{.490\linewidth}\small
        \makeatletter\let\@latex@error\@gobble\makeatother
        \SetInd{0.25em}{0.4em}
        \begin{algorithm}[H]\caption{(Predictive) OMD}\label{algo:predictive omd}
            \DontPrintSemicolon
                                $\vec{z}^0 \in \cD$ such that $\nabla\regu(\vec{z}^0) = \vec{0}$\;\label{line:omd setup}\vspace{.1mm}
            \Hline{}
            \Fn{\normalfont\textsc{NextStrategy}($\vec{m}^{t}$)}{
                \Comment{\color{commentcolor} Set $\vec{m}^t = \vec{0}$ for non-predictive version}\vspace{1mm}
                \Return{$\displaystyle\argmin_{\xhat \in \cD} \mleft\{\langle \vec{m}^t, \xhat\rangle + \frac{1}{\eta}\div{\xhat}{\vec{z}^{t-1}}\mright\}$}\hspace*{-3cm}\;\label{line:omd next xt}\vspace{1mm}
            }
            \Hline{}
            \Fn{\normalfont\textsc{ObserveLoss}($\vec{\ell}^{t}$)}{
                $\displaystyle\vec{z}^{t} \gets \argmin_{\zhat \in \cD}\mleft\{\langle \vec{\ell}^t, \zhat\rangle + \frac{1}{\eta}\div{\zhat}{\vec{z}^{t-1}}\mright\}$\;\label{line:omd next zt}
            }
        \end{algorithm}
    \end{minipage}
\end{figure*}

%% file: text/blackwell_approachability.tex
\section{Blackwell Approachability}

\emph{Blackwell approachability}~\citep{Blackwell56:analog} generalizes the problem of playing a repeated two-player game to games whose utilites are vectors instead of scalars. In a Blackwell approachability game, at all times $t$, two players interact in this order: first, Player 1 selects an action $\vec{x}^t \in \cX$; then, Player 2 selects an action $\vec{y}^t \in \cY$; finally, Player 1 incurs the vector-valued payoff $\vec{u}(\vec{x}^t, \vec{y}^t) \in \bbR^d$, where $\vec{u}$ is a biaffine function. The sets $\cX,\cY$ of player actions are assumed to be compact convex sets. Player 1's objective is to guarantee that the average payoff converges to some desired closed convex \emph{target set} $S \subseteq \bbR^d$. Formally, given target set $S \subseteq \bbR^d$, Player 1's goal is to pick actions $\vec{x}^1, \vec{x}^2, \ldots \in \cX$ such that no matter the actions $\vec{y}^1, \vec{y}^2, \ldots \in \cY$ played by Player 2,
\begin{equation}\label{eq:blackwell goal}
    \min_{\hat{\vec{s}} \in S}\ \mleft\|\hat{\vec{s}} - \frac{1}{T}\sum_{t=1}^T \vec{u}(\vec{x}^t, \vec{y}^t)\mright\|_2 \to 0\quad\text{as}\quad T\to\infty.
\end{equation}

A central concept in the theory of Blackwell approachability is the following.
\begin{definition}[Approachable halfspace, forcing function]\label{def:forcing action}
  Let $(\cX, \cY, \vec{u}(\cdot, \cdot), S)$ be a Blackwell approachability game as described above and let $H \subseteq \bbR^d$ be a halfspace, that is, a set of the form $H = \{\vec{x} \in \bbR^d : \vec{a}^{\!\top}\vec{x} \le b\}$ for some $\vec{a} \in \bbR^d, b \in \bbR$. The halfspace $H$ is said to be \emph{forceable} if there exists a strategy of Player 1 that guarantees that the payoff is in $H$ no matter the actions played by Player 2. In symbols, $H$ is forceable if there exists $\vec{x}^* \in \cX$ such that for all $\vec{y}\in \cY$, $\vec{u}(\vec{x}^*, \vec{y}) \in H$. When this is the case, we call action $\vec{x}^*$ a \emph{forcing action} for $H$.
\end{definition}

Blackwell's \emph{approachability theorem}~\citep{Blackwell56:analog} states that goal \eqref{eq:blackwell goal} can be attained if and only if all halfspaces $H \supseteq S$ are forceable.
Blackwell approachability has a number of applications and connections to other problems in the online learning and game theory literature (e.g.,~\citep{Blackwell54:Controlled,Foster99:Proof,Hart00:Simple}).
%

In this paper we leverage the Blackwell approachability formalism to draw new connections between FTRL and OMD with RM and \rmp{}, respectively. 
We also introduce predictive Blackwell approachability, and show that it can be used to develop new state-of-the-art algorithms for simplex domains and imperfect-information extensive-form zero-sum games.


%% file: text/rm_to_approachability.tex
\section{From Online Linear Optimization to Blackwell Approachability}

\citet{Abernethy11:Blackwell} showed that it is always possible to convert a regret minimizer into an algorithm for a Blackwell approachability game (that is, an algorithm that chooses actions $\vec{x}^t$ at all times $t$ in such a way that goal~\eqref{eq:blackwell goal} holds no matter the actions $\vec{y}^1, \vec{y}^2, \dots$ played by the opponent).%
\footnote{Gordon's Lagrangian Hedging \citep{Gordon05:NoRegret,Gordon07:NoRegret}
        partially overlaps with the construction by \citet{Abernethy11:Blackwell}. 
        We did not investigate to what extent the \emph{predictive}
        point of view we adopted in the paper could apply to Gordon's result.
}

 In this section, we slightly extend their constructive proof by allowing more flexibility in the choice of the domain of the regret minimizer. This extra flexibility will be needed to show that RM and \rmp{} can be obtained directly from FTRL and OMD, respectively.

We start from the case where the target set in the Blackwell approachability game is a closed convex cone $C \subseteq \bbR^n$. As \cref{prop:olo to approachability} shows, \cref{algo:olo to approachability} provides a way of playing the Blackwell approachability game that guarantees that~\eqref{eq:blackwell goal} is satisfied (the proof is in \cref{app:proofs olo to approachability} in the full version of the paper). In broad strokes, \cref{algo:olo to approachability} works as follows (see also \cref{fig:olo to approachability pictorial}): the regret minimizer has as its decision space the polar cone to $C$ (or a subset thereof), and its decision is used as the normal vector in choosing a halfspace to force.  At time t, the algorithm plays a forcing action $\vec{x}^t$ for the halfspace $H_t$ induced by the last decision $\vec{\theta}^t$ output by the OLO oracle $\cL$. Then, $\cL$ incurs the loss $-\vec{u}(\vec{x}^t, \vec{y}^t)$, where $\vec{u}$ is the payoff function of the Blackwell approachability game.

\begin{figure}[t]
    \begin{minipage}[b]{8.8cm}\small
        \makeatletter\let\@latex@error\@gobble\makeatother
        \scalebox{.95}{\begin{algorithm}[H]\caption{From OLO to (predictive) approachability\!\!\!\!\!\!\!}
            \label{algo:olo to approachability}
            \DontPrintSemicolon
            \KwData{$\cD \subseteq \bbR^n$ convex and closed, s.t. $\cK \defeq C^\circ \cap \bbB^n_2 \subseteq \cD \subseteq C^\circ $\hspace*{-2cm}\\
                \hspace{0.9cm}$\cL$ online linear optimization algorithm for domain $\cD$\!\!\!}
            \BlankLine
            \Fn{\normalfont\textsc{NextStrategy}($\vec{v}^{t}$)}{
                \Comment{\color{commentcolor} Set $\vec{v}^t = \vec{0}$ for non-predictive version}\vspace{1mm}
                $\vec{\theta}^t \gets \cL\textsc{.NextStrategy($-\vec{v}^t$)}$\;\label{line:pass in prediction}
                \Return{$\vec{x}^t$ forcing action for $H^t \defeq \{\vec{x}: \langle\vec{\theta}^t), \vec{x} \rangle \le 0\}$}\;            }
            \Hline{}
            \Fn{\normalfont\textsc{ReceivePayoff}($\vec{u}(\vec{x}^{t}, \vec{y}^t)$)}{
                $\cL\textsc{.ObserveLoss}(-\vec{u}(\vec{x}^t, \vec{y}^t))$\;\label{line:pass in loss}
            }
        \end{algorithm}}
    \end{minipage}
\end{figure}
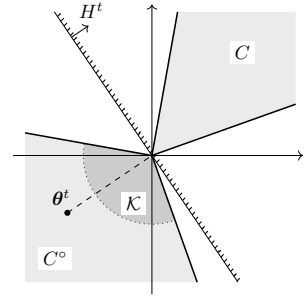
\begin{figure}[h]
    \centering
        \scalebox{.73}{\begin{tikzpicture}[scale=1.1,baseline=0]
            \fill[black!8!white] (0.5*.85, 2.8*.85) -- (0, 0) -- (2.8*.85, 1*.85) -- (2.8*.85, 2.8*.85) -- cycle;
            \fill[black!8!white] (-2.8*.75, 0.5*.75) -- (0, 0) -- (1*.75, -2.8*.75) -- (-2.8*.75,-2.8*.75) -- cycle;
            \draw[black,dotted,fill=black!20!white] (0,0) -- (-2.8*.4, 0.5*.4) arc (170:289.3:1.14) -- cycle;

            \draw[thick] (-2.8*.75, 0.5*.75) -- (0, 0) -- (1*.75, -2.8*.75);
            \draw[thick] (0.5*.85, 2.8*.85) -- (0, 0) -- (2.8*.85, 1*.85);

            \draw[] (1.9*.75, -2.8*.75) -- (-1.9*.85, 2.8*.85);
            \fill[thick,pattern=north east lines] (2.0*.75, -2.8*.75) -- (1.9*.75, -2.8*.75) -- (-1.9*.85, 2.8*.85) -- (-1.8*.85, 2.8*.85) -- cycle;
            \draw[->] (-1.9*.7,2.8*.7) -- ++(2.8*.1,1.9*.1);

            \fill (-2.8*.5,-1.9*.5) circle (.05);
            \draw[dashed] (-2.8*.5,-1.9*.5) -- (0,0);

            \node[fill=white,inner sep=2.5] at (1.5,1.7) {$C$};
            \node[fill=white,inner sep=2.5] at (-1.6,-1.7) {$C^\circ$};
            \node[fill=white,inner sep=2.5] at (-.3,-0.8) {$\cK$};
            \node[] at (-1.5,-0.7) {$\vec{\theta}^t$};
            \node[inner sep=2.5] at (-1.0,2.4) {$H^t$};

            \draw[thin,black,->] (-2.3, 0) -- (2.5, 0);
            \draw[thin,black,->] (0, -2.3) -- (0, 2.5);
        \end{tikzpicture}}
    \vspace{-1mm}
    \caption{Pictorial depiction of \cref{algo:olo to approachability}'s inner working: at all times $t$, the algorithm plays a forcing action for the halfspace $H^t$ induced by the last decision output by $\cL$.}
    \label{fig:olo to approachability pictorial}
    \vspace{-1mm}
\end{figure}

\begin{restatable}{proposition}{propolotoapproachability}\label{prop:olo to approachability}
    Let $(\cX, \cY, \vec{u}(\cdot, \cdot), C)$ be an approachability game, where $C \subseteq \bbR^n$ is a closed convex cone, such that each halfspace $H \supseteq C$ is approachable (\cref{def:forcing action}). Let $\cK \defeq C^\circ \cap \bbB_2^n$, where $C^\circ = \{\vec{x} \in \bbR^n : \langle\vec{x},\vec{y}\rangle \le 0\ \forall\vec{y}\in C\}$ denotes the \emph{polar cone} to $C$ and $\bbB_2^n \defeq \{\vec{x} \in \bbR^n : \|\vec{x}\|_2 \le 1\}$ is the unit ball. Finally, let $\cL$ be an oracle for the OLO problem (for example, the FTRL or OMD algorithm) whose domain of decisions is any closed convex set $\cD$, such that $\cK \subseteq \cD \subseteq C^\circ$. Then, at all times $T$, the distance between the average payoff cumulated by \cref{algo:olo to approachability} and the target cone $C$ is upper bounded as
    \[
        \min_{\hat{\vec{s}}\in C}\ \mleft\| \hat{\vec{s}} - \frac{1}{T}\sum_{t=1}^T \vec{u}(\vec{x}^t, \vec{y}^t)\mright\|_2 \le \frac{1}{T}\max_{\xhat \in \cK} R_\cL^T(\xhat),
    \]
    where $R_\cL^T(\xhat)$ is the regret cumulated by $\cL$ up to time $T$ compared to always playing $\hat{\vec{x}}\in\cK$.
\end{restatable}

As $\cK$ is compact, by virtue of $\cL$ being a regret minimizer, $\nicefrac{1}{T} \cdot \max_{\hat{\vec{x}} \in \cK}R^T(\hat{\vec{x}})\to 0$ as $T \to \infty$, \cref{algo:olo to approachability} satisfies the Blackwell approachability goal \eqref{eq:blackwell goal}.
The fact that \cref{prop:olo to approachability} applies only to conic target sets does not limit its applicability. Indeed, \citet{Abernethy11:Blackwell} showed that any Blackwell approachability game with a non-conic target set can be efficiently transformed to another one with a conic target set. In this paper, we only need to focus on conic target sets.

The construction by \citet{Abernethy11:Blackwell} coincides with \cref{prop:olo to approachability} in the special case where the domain $\cD$ is set to $\cD = \cK$. In the next section, we will need our added flexibility in the choice of $\cD$: in order to establish the connection between \rmp{} and OMD, it is necessary to set $\cD = C^\circ \neq \cK$.

%% file: text/rm_and_rmplus.tex
\section{Connecting FTRL, OMD with RM, \rmp{}}\label{sec:ftrl omd rm rmp}
\input{text/rm_rmp_algos}

Constructing a regret minimizer for a simplex domain $\Delta^n \defeq \{\vec{x} \in \bbR_{\ge 0} : \|\vec{x}\|_1=1\}$ can be reduced to constructing an algorithm for a particular Blackwell approachability game $\Gamma \defeq (\Delta^n, \bbR^n, \vec{u}(\cdot, \cdot), \bbR_{\le 0}^n)$ that we now describe~\cite{Hart00:Simple}. For all $i\in\{1,\dots,n\}$, the $i$-th component of the vector-valued payoff function $\vec{u}$ measures the change in regret incurred at time $t$, compared to always playing the $i$-th vertex $\vec{e}_i$ of the simplex. Formally, $\vec{u}: \Delta^n \times \bbR^n \to \bbR^n$ is defined as
\begin{equation}\label{eq:blackwell simplex utility}
    \vec{u}(\vec{x}^t, \vec{\ell}^t) = \langle \vec{\ell}^t, \vec{x}^t\rangle \vec{1} - \vec{\ell}^t, 
\end{equation}
where $\vec{1}$ is the $n$-dimensional vector whose components are all $1$.
It is known that $\Gamma$ is such that the halfspace 
$
    H_{\vec{a}} \defeq \{\vec{x} \in \bbR^n : \langle\vec{x}, \vec{a}\rangle \le 0 \} \supseteq \bbR^n_{\le 0}
$
is forceable (\cref{def:forcing action}) for all $\vec{a} \in \bbR^n_{\ge 0}$. A forcing action for $H_{\vec{a}}$ is given by  $\vec{g}(\vec{a}) \defeq \vec{a} / \|\vec{a}\|_1 \in \Delta^n$ when $\vec{a} \neq \vec{0}$; when $\vec{a} = \vec{0}$, any $\vec{x}\in\Delta^n$ is a forcing action.
The following is known.
\begin{restatable}{lemma}{lemsimplextoblackwell}
The regret
$
    R^T(\xhat) = \frac{1}{T}\sum_{t=1}^T \langle \vec{\ell}^t, \vec{x}^t - \xhat\rangle
$
cumulated up to any time $T$ by the decisions $\vec{x}^1, \dots,\vec{x}^T \in \Delta^n$
compared to any $\xhat \in \Delta^n$
is related to the distance of the average Blackwell payoff from the target cone $\bbR_{\le 0}^n$ as
\begin{align}
  \frac{1}{T} R^T(\xhat) \le \min_{\hat{\vec{s}}\in\bbR_{\le 0}^n}\mleft\|\hat{\vec{s}} - \frac{1}{T}\sum_{t=1}^T \vec{u}(\vec{x}^t, \vec{\ell}^t)\mright\|_2.
\label{eq:simplex to blackwell}
\end{align}
So, a strategy for the Blackwell approachability game $\Gamma$ is a regret-minimizing
strategy for the simplex domain $\Delta^n$.
\end{restatable}

When the approachability game $\Gamma$ is solved by means of the constructive proof of Blackwell's approachability theorem~\citep{Blackwell56:analog}, one recovers a particular regret minimizer for the domain $\Delta^n$ known as the \emph{regret matching (RM)} algorithm~\cite{Hart00:Simple}. The same cannot be said for the closely related \rmp{} algorithm~\cite{Tammelin14:Solving}, which converges significantly faster in practice than RM, as has been reported many times.

We now uncover deep and surprising connections between RM, \rmp{} and the OLO algorithms FTRL, OMD by solving $\Gamma$ using \cref{algo:olo to approachability}.
  Let $\cL^\text{ftrl}_\eta$ be the FTRL algorithm instantiated over the conic domain $\cD = \bbR^n_{\ge 0}$ with the $1$-strongly convex regularizer $\regu(\vec{x}) = \nicefrac{1}{2}\,\|\vec{x}\|^2_2$ and an arbitrary step size parameter $\eta$. Similarly, let $\cL^\text{omd}_\eta$ be the OMD algorithm instantiated over the same domain $\cD = \bbR^n_{\ge 0}$ with the same convex regularizer $\regu(\vec{x}) = \nicefrac{1}{2}\,\|\vec{x}\|^2_2$. Since $\bbR_{\ge 0}^n = (\bbR^n_{\le 0})^\circ$, $\cD$ satisfies the requirements of \cref{prop:olo to approachability}. So, $\cL^\text{ftrl}_\eta$ and $\cL^\text{omd}_\eta$ can be plugged into \cref{algo:olo to approachability} to compute a strategy for the Blackwell approachability game $\Gamma$. When that is done, the following can be shown (all proofs for this section are in \cref{app:proofs rm rmp} in the full version of the paper).

\begin{restatable}[FTRL reduces to RM]{theorem}{thmrmisftrl}\label{thm:rm is ftrl}
  For all $\eta > 0$, when \cref{algo:olo to approachability} is set up with $\cD = \bbR^n_{\ge 0}$ and regret minimizer $\cL^\text{\normalfont ftrl}_\eta$ to play $\Gamma$, it produces the same iterates as the RM algorithm.
\end{restatable}
\begin{restatable}[OMD reduces to \rmp{}]{theorem}{thmrmpisomd}\label{thm:rmp is omd}
  For all $\eta \!>\! 0$, when \cref{algo:olo to approachability} is set up with $\cD \!=\! \bbR^n_{\ge 0}$ and regret minimizer $\cL^\text{\normalfont omd}_\eta$ to play $\Gamma$, it produces the same iterates as the \rmp{} algorithm.
\end{restatable}

Pseudocode for RM and \rmp{} is given in \cref{algo:prm,algo:prmp} (when $\vec{m}^t = \vec{0}$). In hindsight, the equivalence between RM and \rmp{} with FTRL and OMD is clear. The computation of $\vec{\theta}^t$ on Line~3 in both PRM and \prmp{} corresponds to the closed-form solution for the minimization problems of Line~4 in FTRL and Line~3 in OMD, respectively, in accordance with Line~2 of \cref{algo:olo to approachability}. Next, Lines~4 and~5 in both PRM and \prmp{} compute the forcing action required in Line~3 of \cref{algo:olo to approachability} using the function $\vec{g}$ defined above. Finally, in accordance with Line~6 of \cref{algo:olo to approachability}, Line~7 of PRM corresponds to Line~6 of FTRL, and Line~7 of \prmp{} to Line~5 of OMD.


%% file: text/rm_rmp_algos.tex
\begin{figure*}\centering
    \begin{minipage}[t]{.490\linewidth}\small
        \makeatletter\let\@latex@error\@gobble\makeatother
        \begin{algorithm}[H]\caption{(Predictive) regret matching}\label{algo:prm}
            \DontPrintSemicolon
            $\vec{r}^0 \gets \vec{0} \in \bbR^n,\ \ \vec{x}^0 \gets \vec{1}/n \in \Delta^n$\;
            \Hline{}
            \Fn{\normalfont\textsc{NextStrategy}($\vec{m}^{t}$)}{
                \Comment{\color{commentcolor} Set $\vec{m}^t = \vec{0}$ for non-predictive version}\vspace{1mm}
                $\displaystyle\vec{\theta}^t \gets [\vec{r}^{t-1} + \langle\vec{m}^t,\vec{x}^{t-1}\rangle \vec{1} - \vec{m}^t]^+$\;
                \textbf{if} $\vec{\theta}^t \neq \vec{0}$ \textbf{return} $\vec{x}^t \gets \vec{\theta}^t \ /\ \|\vec{\theta}^t\|_1$\;
                \textbf{else} \hspace{0.715cm}\textbf{return} $\vec{x}^t \gets $ arbitrary point in $\Delta^{\!n}$\hspace*{-1cm}\;
            }
            \Hline{}
            \Fn{\normalfont\textsc{ObserveLoss}($\vec{\ell}^{t}$)}{
                $\displaystyle\vec{r}^t \gets \vec{r}^{t-1} + \langle\vec{\ell}^t,\vec{x}^t\rangle \vec{1} - \vec{\ell}^t$\;\vspace{.5mm}
            }
        \end{algorithm}
    \end{minipage}
    \hfill
    \begin{minipage}[t]{.490\linewidth}\small
        \makeatletter\let\@latex@error\@gobble\makeatother
        \begin{algorithm}[H]\caption{(Predictive) regret matching$^+$\!\!\!}\label{algo:prmp}
            \DontPrintSemicolon
            $\vec{z}^0 \gets \vec{0} \in \bbR^n,\ \ \vec{x}^0 \gets \vec{1}/n \in \Delta^n$\;
            \Hline{}
            \Fn{\normalfont\textsc{NextStrategy}($\vec{m}^{t}$)}{
                \Comment{\color{commentcolor} Set $\vec{m}^t = \vec{0}$ for non-predictive version}\vspace{1mm}
                $\displaystyle\vec{\theta}^t \gets [\vec{z}^{t-1} + \langle\vec{m}^t,\vec{x}^{t-1}\rangle \vec{1} - \vec{m}^t]^+$\;
                \textbf{if} $\vec{\theta}^t \neq \vec{0}$ \textbf{return} $\vec{x}^t \gets \vec{\theta}^t \ /\ \|\vec{\theta}^t\|_1$\;
                \textbf{else} \hspace{0.715cm}\textbf{return} $\vec{x}^t \gets $ arbitrary point in $\Delta^{\!n}$\hspace*{-1cm}\;
            }
            \Hline{}
            \Fn{\normalfont\textsc{ObserveLoss}($\vec{\ell}^{t}$)}{
                $\displaystyle\vec{z}^t \gets [\vec{z}^{t-1} + \langle\vec{\ell}^t,\vec{x}^t\rangle \vec{1} - \vec{\ell}^t]^+$\;
            }
        \end{algorithm}
    \end{minipage}
    \vspace{-4mm}
\end{figure*}

%% file: text/optimistic_approachability.tex
\section{Predictive Blackwell Approachability, and Predictive RM and \rmp{}}

It is natural to wonder whether it is possible to devise an algorithm for Blackwell approachability games that is able to guarantee faster convergence to the target set when good predictions of the next vector payoff are available. We call this setup \emph{predictive Blackwell approachability}. We answer the question in the positive by leveraging \cref{prop:olo to approachability}. Since the loss incurred by the regret minimizer is $\vec{\ell}^t \defeq -\vec{u}(\vec{x}^t, \vec{y}^t)$ (\cref{line:pass in loss} in \cref{algo:olo to approachability}), any prediction $\vec{v}^{t}$ of the  payoff $\vec{u}(\vec{x}^t, \vec{y}^t)$ is naturally a prediction about the next loss incurred by the underlying regret minimizer $\cL$ used in \cref{algo:olo to approachability}. Hence, as long as the prediction is propagated as in \cref{line:pass in prediction} in \cref{algo:olo to approachability}, \cref{prop:olo to approachability} holds verbatim. In particular, we prove the following. All proofs for this section are in \cref{app:prm prmp} in the full version of the paper.

\begin{restatable}{proposition}{proppredictiveblackwell}
  Let $(\cX,\cY,\vec{u}(\cdot,\cdot),S)$ be a Blackwell approachability game, where every halfspace $H \supseteq S$ is approachable (\cref{def:forcing action}). For all $T$, given predictions $\vec{v}^t$ of the payoff vectors, there exist algorithms for playing the game (that is, pick $\vec{x}^t \in \cX$ at all $t$) that guarantee
  \begin{align*}
        &\min_{\hat{\vec{s}}\in S}\mleft\| \hat{\vec{s}}\!-\! \frac{1}{T}\!\sum_{t=1}^T \!\vec{u}(\vec{x}^t\!\!, \vec{y}^t)\mright\|_2
        \!\!\!\!\le\!\!\frac{1}{\sqrt{T}}\!\mleft(\!\! 1 \!+\! \frac{2}{T}\! \sum_{t=1}^T \!\|\vec{u}(\vec{x}^t\!\!,\! \vec{y}^t) \!-\! \vec{v}^t\|_2^2\!\mright)\!.
  \end{align*}
\end{restatable}

We now focus on how predictive Blackwell approachability ties into our discussion of RM and \rmp{}.
In \cref{sec:ftrl omd rm rmp} we showed that when \cref{algo:olo to approachability} is used in conjunction with FTRL and OMD on the Blackwell approachability game $\Gamma$ of \cref{sec:ftrl omd rm rmp}, the iterates coincide with those of RM and \rmp{}, respectively.
In the rest of this section we investigate the use of \emph{predictive} FTRL and \emph{predictive} OMD in that framework.
Specifically, we use predictive FTRL and preditictive OMD as the regret minimizers to solve the Blackwell approachability game introduced in
\cref{sec:ftrl omd rm rmp}, and coin the resulting predictive regret minimization algorithms for simplex domains \emph{predictive regret matching (PRM)} and \emph{predictive regret matching$^+$ (\prmp{})}, respectively.
Ideally, starting from the prediction $\vec{m}^t$ of the next loss, we would want the prediction $\vec{v}^t$ of the next utility in the equivalent Blackwell game $\Gamma$ (\cref{sec:ftrl omd rm rmp}) to be $\vec{v}^t = \langle \vec{m}^t, \vec{x}^{t} \rangle \vec{1} - \vec{m}^t$ to maintain symmetry with \eqref{eq:blackwell simplex utility}. However, $\vec{v}^t$ is computed before $\vec{x}^t$ is computed, and $\vec{x}^t$ depends on $\vec{v}^t$, so the previous expression requires the computation of a fixed point. To sidestep this issue, we let
\[
    \vec{v}^t \defeq \langle \vec{m}^t, \vec{x}^{t-1} \rangle \vec{1} - \vec{m}^t
\]
instead.
We give pseudocode for PRM and \prmp{} as \cref{algo:prm,algo:prmp}.
%
In the rest of this section, we discuss formal guarantees for PRM and \prmp{}.
\begin{restatable}[Correctness of PRM, \prmp{}]{theorem}{thmprmprmp}\label{thm:prm prmp bound}
 Let $\cL^\text{\normalfont ftrl*}_\eta$ and $\cL^\text{\normalfont omd*}_\eta$ denote the predictive FTRL and predictive OMD algorithms instantiated with the same choice of regularizer and domain as in \cref{sec:ftrl omd rm rmp}, and predictions $\vec{v}^t$ as defined above for the Blackwell approachability game $\Gamma$. For all $\eta > 0$, when \cref{algo:olo to approachability} is set up with $\cD = \bbR_{\ge 0}^n$, the regret minimizer $\cL^\text{\normalfont ftrl*}_\eta$ (resp., $\cL^\text{\normalfont omd*}_\eta$) to play $\Gamma$, it produces the same iterates as the PRM (resp., \prmp{}) algorithm. Furthermore, PRM and \prmp{} are regret minimizer for the domain $\Delta^n$, and at all times $T$ satisfy the regret bound
 \begin{align*}
     R^T(\hat{\vec{x}}) &\le \sqrt{2} \mleft(\sum_{t=1}^T \|\vec{u}(\vec{x}^t\!, \vec{\ell}^t) - \vec{v}^t\|_2^2\mright)^{\!\!\!1/2}.
\end{align*}
\end{restatable}

\noindent At a high level, the main insight behind the regret bound of \cref{thm:prm prmp bound} is to combine
\cref{prop:olo to approachability} with the guarantees of predictive FTRL and
predictive OMD (\cref{prop:oco bound}). In particular, combining \eqref{eq:simplex to blackwell} with \cref{prop:olo to approachability}, we find that the regret $R^T$ cumulated by the strategies $\vec{x}^1,\dots,\vec{x}^T$ produced up to time $T$ by PRM and \prmp{} satisfies
\begin{equation}\label{eq:bound1}
    \frac{1}{T}\max_{\xhat\in\Delta^n} R^T(\xhat) \le \frac{1}{T}\max_{\hat{\vec{x}} \in \bbR_{\ge 0}^n \cap \bbB_2^n} R_\cL^T(\xhat),
\end{equation}
where $\cL = \cL_\eta^\text{ftrl*}$ for PRM and $\cL = \cL_\eta^\text{omd*}$ for \prmp{}. Since the domain of the maximization on the
right hand side is a subset of the domain $\cD = \bbR^n_{\ge 0}$ of $\cL$, the bound in
\cref{prop:oco bound} holds, and in particular
\begin{align}
    \max_{\xhat\in\Delta^n}\!R^T(\xhat)\! &\le\! \max_{\xhat \in \bbR_{\ge 0}^n \cap \bbB_2^n}\!\mleft\{\!\frac{\|\xhat\|_2^2}{2\eta} \!+\! \eta\!\sum_{t=1}^T \|\vec{u}(\vec{x}^t\!\!,\vec{\ell}^t) \!-\! \vec{v}^t\|^2_2\!\mright\}\nonumber\\
                                       &\le\mleft(\frac{1}{2\eta} + \eta\sum_{t=1}^T \|\vec{u}(\vec{x}^t, \vec{\ell}^t) - \vec{v}^t\|^2_2\mright),
\label{eq:bound2}
\end{align}
where in the first inequality we used the fact that $\regu(\xhat) = \|\xhat\|_2^2/2$ by construction and in the second inequality we used the definition of unit ball $\bbB_2^n$.
Finally, using the fact that the iterates produced by PRM and \prmp{} do not depend on the chosen step size $\eta > 0$ (first part of \cref{thm:prm prmp bound}),
we conclude that \eqref{eq:bound2} must hold true for any $\eta > 0$, and so in particular also the $\eta > 0$ that minimizes the right hand side:
\begin{align*}
    \max_{\xhat\in\Delta^n} R^T(\xhat) &\le \inf_{\eta > 0}\mleft\{\frac{1}{2\eta} + \eta\sum_{t=1}^T \|\vec{u}(\vec{x}^t, \vec{\ell}^t) - \vec{v}^t\|^2_2\mright\}\\
                                       &= \sqrt{2}\mleft(\sum_{t=1}^T \|\vec{u}(\vec{x}^t, \vec{\ell}^t) - \vec{v}^2\|_2^2\mright)^{\!\!\!1/2}.
\end{align*}

%% file: text/experiments.tex
\section{Experiments}

We conduct experiments on solving two-player zero-sum games. As mentioned previously, for EFGs the CFR framework is used for decomposing regrets into local regret minimization problems at each simplex corresponding to a decision point in the game~\citep{Zinkevich07:Regret,Farina19:Online}, and we do the same.
However, as the regret minimizer for each local decision point, we use \prmp{} instead of RM. In addition, we apply two heuristics that usually lead to better practical performance: we use quadratic averaging of the strategy iterates, that is, we average the sequence-form strategies $\vec{x}^1, \dots, \vec{x}^T$ using the formula
$
    \frac{6}{T(T+1)(2T+1)} \sum_{t=1}^T t^2 \vec{x}^t,
$
and we use the \emph{alternating updates} scheme.
We call this algorithm \pcfrp{}. We compare \pcfrp{} to the prior state-of-the-art CFR variants: \cfrp{}~\citep{Tammelin14:Solving}, \emph{Discounted CFR (DCFR)} with its recommended parameters~\citep{Brown19:Solving}, and \emph{Linear CFR (LCFR)}~\citep{Brown19:Solving}.

\input{text/experiments_figure}

We conduct the experiments on common benchmark games. We show results on seven games in the main body of the paper. An additional 11 games are shown in the appendix of the full version of the paper. The experiments shown in the main body are representative of those in the appendix. A description of all the games is in \cref{app:games} in the full version of the paper, and the results are shown in \cref{fig:plots}. The x-axis shows the number of iterations of each algorithm. Every algorithm pays almost exactly the same cost per iteration, since the predictions require only one additional thresholding step in \pcfrp.
For each game, the top plot shows on the y-axis the Nash gap, while the bottom plot shows the accuracy in our predictions of the regret vector, measured as the average $\ell_2$ norm of the difference between the actual loss $\vec{\ell}^t$ received and its prediction $\vec{m}^t$ across all regret minimizers at all decision points in the game. For all non-predictive algorithms (\cfrp{}, LCFR, and DCFR), we let $\vec{m}^t = \vec{0}$. For our predictive algorithm, we set $\vec{m^t} = \vec{\ell}^{t-1}$ at all times $t \ge 2$ and $\vec{m}^1 = \vec{0}$. Both y-axes are in log scale.
On Battleship and Pursuit-evasion, \pcfrp{} is faster than the other algorithms by 3-6 orders of magnitude already after 500 iterations, and around 10 orders of magnitude after 2000 iterations. On Goofspiel, \pcfrp{} is also significantly faster than the other algorithms, by 0.5-1 order of magnitude. Finally, in the River endgame, our only poker experiment here, \pcfrp{} is slightly faster than \cfrp, but slower than DCFR.
Finally, \prmp{} converges very rapidly on the \emph{smallmatrix} game, a 2-by-2 matrix game where \cfrp{} and other RM-based methods converge at a rate slower than $T^{-1}$~\cite{Farina19:Optimistic}.
Across all non-poker games in the appendix, we also find that \pcfrp{} beats the other algorithms, often by several orders of magnitude. We conclude that \pcfrp{} seems to be the fastest method for solving non-poker EFGs. The only exception to the non-poker-game empirical rule is Liar's Dice (game \ref{game:ld_new}), where our predictive method performs comparably to DCFR. In the appendix, we also test \cfrp{} with quadratic averaging (as opposed to the linear averaging that \cfrp{} normally uses). This does not change any of our conclusions, except that for Liar's Dice, \cfrp{} performs comparably to DCFR and \pcfrp{} when using quadratic averaging (in fact, quadratic averaging hurts \cfrp{} in every game except poker and Liar's Dice).

We tested on three poker games, the River endgame shown here (which is a real endgame encountered by the \emph{Libratus} AI~\citep{Brown17:Superhuman} in the man-machine ``Brains vs. Artificial Intelligence: Upping the Ante'' competition), as well as Kuhn and Leduc poker in the appendix. On Kuhn poker, \pcfrp{} is extremely fast and the fastest of the algorithms. That game is known to be significantly easier than deeper EFGs for predictive algorithms~\citep{Farina19:Optimistic}. On Leduc poker as well as the River endgame, the predictions in \pcfrp{} do not seem to help as much as in other games. On the River endgame, the performance is essentially the same as that of \cfrp. On Leduc poker, it leads to a small speedup over \cfrp. On both of those games, DCFR is fastest. In contrast, DCFR actually performs worse than \cfrp{} in our non-poker experiments, though it is sometimes on par with \cfrp. In the appendix, where we try quadratic averaging in \cfrp{}, we find that for poker games this does speed up \cfrp{}, and allows it to be slightly faster than \pcfrp{} on the River endgame and Leduc poker.
We conclude that \pcfrp{} is much faster than \cfrp{} and DCFR on non-poker games, whereas on poker games DCFR is the fastest. 

The convergence rate of \pcfrp{} is closely related to how good the predictions $\vec{m}^t$ of $\vec{\ell}^t$ are. On Battleship and Pursuit-evasion, the predictions become extremely accurate very rapidly, and \pcfrp{} converges at an extremely fast rate. On Goofspiel, the predictions are fairly accurate (the error is of the order $10^{-5}$) and \pcfrp{} is still significantly faster than the other algorithms. On the River endgame, the average prediction error is of the order  $10^{-3}$, and \pcfrp{} performs on par with \cfrp, and slower than DCFR. Similar trends prevail in the experiments in the appendix.
Additional experimental insights are described in the appendix. 

%% file: text/experiments_figure.tex
\begin{figure*}[!t]
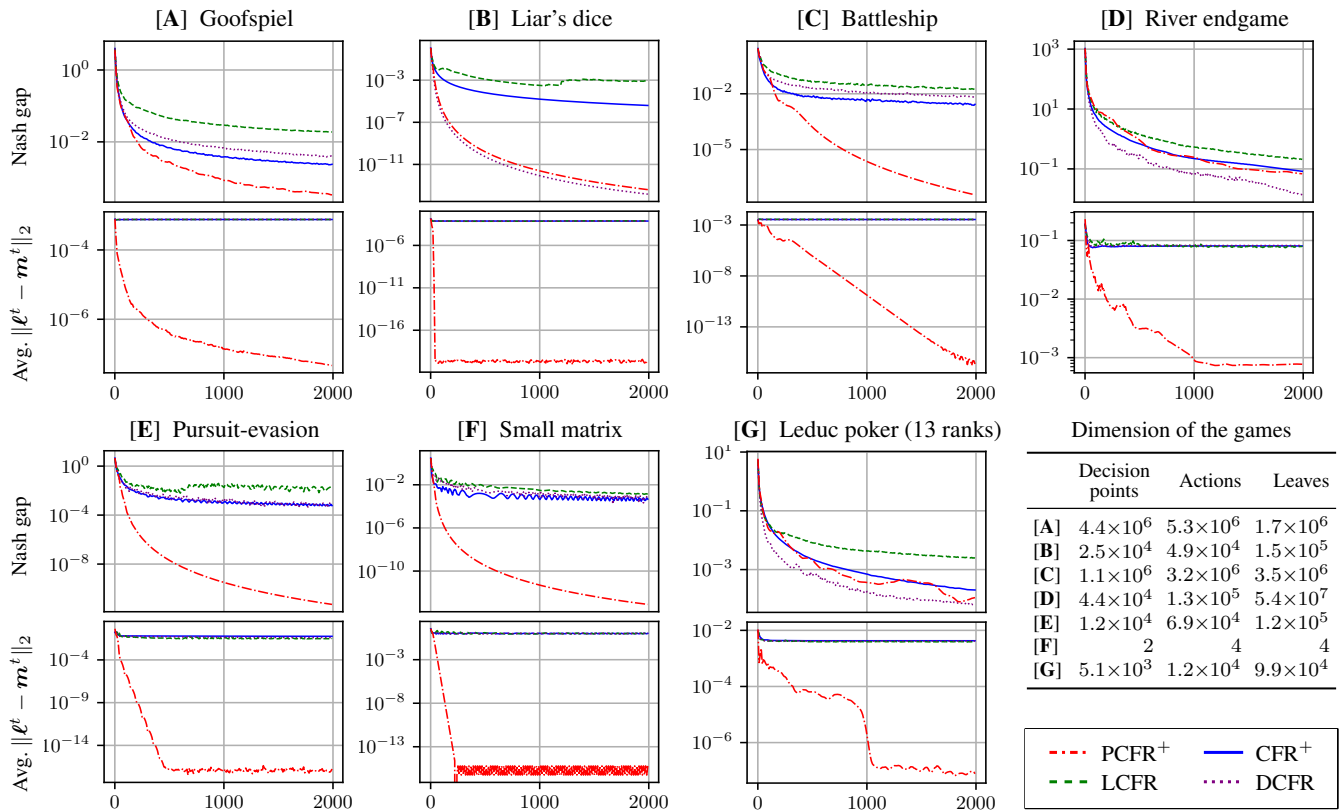
\raggedright
    \addplotA{goof5}{Goofspiel}%
    \hspace*{-.5mm}\addplotA*{ld_new}{Liar's dice}%
    \hspace*{-.5mm}\addplotA*{bs23_4turns}{Battleship}%
    \hspace*{-.5mm}\addplotA*{river_endgame7}{River endgame}\\
    \addplotA{search6}{Pursuit-evasion}%
    \hspace*{-.5mm}\addplotA*{sm}{Small matrix}%
    \hspace*{-.5mm}\addplotA*{leduc13}{Leduc poker (13 ranks)}%
    \hfill\scalebox{.95}{\begin{minipage}[t]{4.4cm}%
      \small%
      \centering Dimension of the games\\[1.08mm]
      {%
        \fontsize{8}{8}\selectfont%
        \setlength{\tabcolsep}{.9mm}%
        \renewcommand{\arraystretch}{1.1}%
        \sisetup{
            scientific-notation=true,
            round-precision=1,
            round-mode=places,
            exponent-product={\mkern-4mu\times\mkern-4.5mu}
        }%
      \begin{tabular}{l>{\raggedleft\let\newline\\\arraybackslash\hspace{0pt}}m{1.1cm}rr}
        \toprule
         & \centering Decision points & Actions & Leaves \\%
        \midrule%
        \ref*{game:goof5}          & \num{4369010} & \num{5332052} & \num{1728000} \\%
        \ref*{game:ld_new}         & \num{24576}   & \num{49142}   & \num{147420} \\
        \ref*{game:bs23_4turns}    & \num{1050723} & \num{3236158} & \num{3487428} \\
        \ref*{game:river_endgame7} & \num{44020}   & \num{129222}  & \num{54200135} \\
        \ref*{game:search6}        & \num{11888}   & \num{69029}   & \num{118514} \\
        \ref*{game:sm}             &         $2$ &          $4$ &          $4$ \\
        \ref*{game:leduc13}        &  \num{5148} &  \num{12014} &  \num{98956} \\
        \bottomrule
      \end{tabular}
      }\\[.4cm]
      \begin{tcolorbox}[
        boxsep=0pt,
        left=1pt,right=1pt,top=5pt,bottom=4pt,
        boxrule=.5pt,
        colback=black!0!white,
        colframe=black,
        arc=0pt
      ]\centering
        \begin{tabular}{ll}
         \linesty{pred,dash dot}~~\pcfrp{}&
         \linesty{pblue}~~\cfrp{}\\[.6mm]
         \linesty{pgreen,dashed}~~LCFR&
         \linesty{ppurple,dotted}~~DCFR\\
        \end{tabular}
      \end{tcolorbox}
    \end{minipage}}\\
    \vspace{-3mm}
    \caption{Performance of \pcfrp, \cfrp, DCFR, and LCFR on five EFGs. In all plots, the x axis is the number of iterations of each algorithm. For each game, the top plot shows that the Nash gap on the y axis (on a log scale), the bottom plot shows and the average prediction error (on a log scale).}
    \label{fig:plots}
    \vspace{-5mm}
\end{figure*}

%% file: text/conclusions.tex
\section{Conclusions and Future Research}

We extended \citet{Abernethy11:Blackwell}'s reduction of Blackwell approachability to regret minimization beyond the compact setting. This extended reduction allowed us to show that FTRL applied to the decision of which halfspace to force in Blackwell approachability is equivalent to the regret matching algorithm. OMD applied to the same problem turned out to be equivalent to \rmp. 
Then, we showed that the predictive variants of FTRL and OMD yield predictive algorithms for Blackwell approachability, as well as predictive variants of RM and \rmp. Combining \prmp{} with CFR, we introduced the \pcfrp{} algorithm for solving EFGs. Experiments across many common benchmark games showed that \pcfrp{} outperforms the prior state-of-the-art algorithms on non-poker games by orders of magnitude.

%
This work also opens future directions.
Can \prmp{} guarantee $T^{-1}$ convergence on matrix games like optimistic FTRL and OMD, or do the less stable updates prevent that?
Can one develop a predictive variant of DCFR, which is faster on poker domains?
Can one combine DCFR and \pcfrp, so DCFR would be faster initially but \pcfrp{} would overtake? If the cross-over point could be approximated, this might yield a best-of-both-worlds algorithm.

%% file: text/appendix_ftrl_omd.tex
\section{Analysis of (Predictive) FTRL}\label{app:proofs ftrl}

In the proof of \cref{prop:omd bound} we will use the following technical lemma (see, e.g, \cite{Farina19:Optimistic}).

\begin{lemma}\label{lem:psi lipschitz}
  Let $\regu : \cD \to \bbR_{\ge 0}$ be a $1$-strongly convex differentiable regularizer with respect to some norm $\|\cdot\|$, and let $\|\cdot\|_*$ be the dual norm to $\|\cdot\|$. Finally, let $\vec{\psi} : \bbR^n \to \cD$ be the function
    \[
        \vec{\psi} : \vec{g} \mapsto \argmin_{\xhat \in \cD}\mleft\{\langle \vec{g}, \xhat \rangle + \frac{1}{\eta} \regu(\xhat)\mright\}.
    \]
    Then, $\vec{\psi}$ is $\eta$-Lipschitz continuous with respect to the dual norm, in the sense that
    \[
        \|\vec{\psi}(\vec{g}) - \vec{\psi}(\vec{g'})\| \le \eta\,\|\vec{g} - \vec{g}'\|_* \quad\forall\vec{g},\vec{g}'\in\bbR^n.
    \]
\end{lemma}

\begin{restatable}{proposition}{propftrlbound}\label{prop:ftrl bound}
    Let $\regu : \cD \to \bbR_{\ge 0}$ be a $1$-strongly regularizer with respect to some norm $\|\cdot\|$, and let $\|\cdot\|_*$ be the dual norm to $\|\cdot\|$. For all $\xhat \in \cD$, all $\eta > 0$, and all times $T$, the regret cumulated by (predictive) FTRL (\cref{algo:predictive ftrl}) compared to any fixed strategy $\hat{\vec{x}} \in \cD$ is bounded as
    \begin{equation}\label{eq:oftrl regret bound}
        R^T(\xhat) \le \frac{\regu(\xhat)}{\eta} + \eta\sum_{t=1}^T \|\vec{\ell}^t - \vec{m}^t\|_*^2 - \frac{1}{4\eta} \sum_{t=1}^{T-1} \|\vec{x}^{t+1} - \vec{x}^t\|^2.
    \end{equation}
\end{restatable}
\begin{proof} We combine several techniques and insights from the original works of \citet{Rakhlin13:Online} and \citet{Syrgkanis15:Fast}.
    Let $\vec{\psi} : \bbR^n \to \cD$ be the function that maps
    \[
        \vec{\psi} : \vec{g} \mapsto \argmin_{\xhat \in \cD}\mleft\{\langle \vec{g}, \xhat \rangle + \frac{1}{\eta} \regu(\xhat)\mright\}.
    \]
    With that notation, at all times $t$, predictive FTRL outputs the decision $\vec{x}^t = \vec{\psi}(\vec{L}^{t-1}+\vec{m}^t)$, where $\vec{L}^{t-1} = \sum_{\tau=1}^{t-1} \vec{\ell}^\tau$. For the purpose of this proof, we also introduce the sequence $\vec{w}^t \defeq \vec{\psi}(\vec{L}^t)$ for $t = 1, 2, \dots$. For any $\xhat \in \cD$,
    \[
        R^T(\xhat) = \sum_{t=1}^T \langle \vec{\ell}^t, \vec{x}^t - \xhat\rangle = %
\underbrace{\sum_{t=1}^T\langle \vec{m}^t, \vec{x}^t - \vec{w}^t \rangle%
            +%
            \langle \vec{\ell}^t, \vec{w}^t - \xhat \rangle}_{\circled{A}}%
            +
            \underbrace{\sum_{t=1}^T\langle \vec{\ell}^t - \vec{m}^t, \vec{x}^t - \vec{w}^t \rangle}_{\circled{B}}%
    \]
    We now bound each of the three terms on the right-hand side:
    \begin{itemize}[leftmargin=8mm,nolistsep,itemsep=1mm]
        \item[\circled{A}]
        A critical observation to bound \circled{A} is the following. Since $\vec{\psi}(\vec{g})$ is a minimizer of $\langle \vec{g},\hat{\vec{x}}\rangle + \frac{1}{\eta}\regu(\hat{\vec{x}})$, then by the fist-order optimality conditions,
        \begin{equation}\label{eq:psi opt}
            \mleft\langle \vec{g} + \frac{1}{\eta}\nabla\regu(\vec{\psi}(\vec{g})),\ \vec{\xi} - \vec{\psi}(\vec{g})\mright\rangle \ge 0 \quad\forall\vec{g}\in\bbR^n, \vec{\xi} \in \cD.
        \end{equation}
        Using the hypothesis on the $1$-strongly convexity of $\regu$ and applying (\ref{eq:psi opt}), for all $\vec{\xi}$ we obtain
        \begin{align*}
            \frac{1}{\eta}\regu(\vec{\xi}) + \langle \vec{g},\vec{\xi}\rangle &\ge \frac{1}{\eta}\regu(\vec{\psi}(\vec{g})) + \langle \vec{g}, \vec{\psi}(\vec{g})\rangle + \mleft\langle \vec{g} + \frac{1}{\eta}\nabla\regu(\vec{\psi}(\vec{g})),\ \vec{\xi} - \vec{\psi}(\vec{g})\mright\rangle + \frac{1}{2\eta}\|\vec{\xi} - \vec{\psi}(\vec{g})\|^2\\
            &\ge \frac{1}{\eta}\regu(\vec{\psi}(\vec{g})) + \langle \vec{g}, \vec{\psi}(\vec{g})\rangle + \frac{1}{2\eta}\|\vec{\xi} - \vec{\psi}(\vec{g})\|^2\numberthis\label{eq:crucial}.
        \end{align*}
        By applying~\eqref{eq:crucial} to the two choices $(\vec{g}, \vec{\xi}) = (\vec{L}^{t-1}, \vec{x}^t), (\vec{L}^{t-1} + \vec{m}^t, \vec{w}^t)$, respectively, we have the two inequalities
        \begin{align*}
            \frac{1}{\eta}\regu(\vec{x}^t) + \langle \vec{L}^{t-1}, \vec{x}^t\rangle &\ge \frac{1}{\eta}\regu(\vec{w}^{t-1}) + \langle \vec{L}^{t-1}, \vec{w}^{t-1}\rangle + \frac{1}{2\eta}\|\vec{x}^t - \vec{w}^{t-1}\|^2\\
            \frac{1}{\eta}\regu(\vec{w}^t) + \langle \vec{L}^{t-1} + \vec{m}^t, \vec{w}^t\rangle &\ge \frac{1}{\eta}\regu(\vec{x}^{t}) + \langle \vec{L}^{t-1} + \vec{m}^t, \vec{x}^{t}\rangle + \frac{1}{2\eta}\|\vec{w}^t - \vec{x}^{t}\|^2.
        \end{align*}
        Summing the two above inequalities and rearranging terms yields
        \begin{align*}
            \langle\vec{m}^t,\vec{x}^t -\vec{w}^t\rangle &\le \frac{1}{\eta}(\regu(\vec{w}^t) - \regu(\vec{w}^{t-1})) + \langle\vec{L}^{t-1}, \vec{w}^t - \vec{w}^{t-1}\rangle-\frac{1}{2\eta}\Big(\|\vec{x}^t - \vec{w}^{t-1}\|^2 + \|\vec{w}^t - \vec{x}^t\|^2\Big).
        \end{align*}
        Summing over $t = 1, \dots, T$ and simplifying telescopic terms,
        \begin{align*}
            &\sum_{t=1}^T\langle\vec{m}^t,\vec{x}^t -\vec{w}^t\rangle \le \frac{1}{\eta}(\regu(\vec{w}^T) - \regu(\vec{w}^{0})) + \sum_{t=1}^T \langle\vec{L}^{t-1}, \vec{w}^t - \vec{w}^{t-1}\rangle - \sum_{t=1}^T\frac{1}{2\eta}\Big(\|\vec{x}^t - \vec{w}^{t-1}\|^2 + \|\vec{w}^t - \vec{x}^t\|^2\Big)\\
            &\hspace{1.5cm}\le \frac{1}{\eta}(\regu(\vec{w}^T) - \regu(\vec{w}^{0})) + \sum_{t=1}^T \langle\vec{L}^{t-1}, \vec{w}^t - \vec{w}^{t-1}\rangle - \sum_{t=1}^{T-1}\frac{1}{2\eta}\Big(\|\vec{x}^{t+1} - \vec{w}^{t}\|^2 + \|\vec{w}^{t} - \vec{x}^{t}\|^2\Big)\\
            &\hspace{1.5cm}\le \frac{1}{\eta}(\regu(\vec{w}^T) - \regu(\vec{w}^{0})) + \sum_{t=1}^T \langle\vec{L}^{t-1}, \vec{w}^t - \vec{w}^{t-1}\rangle - \sum_{t=1}^{T-1}\frac{1}{4\eta}\|\vec{x}^{t+1} - \vec{x}^{t}\|^2,
        \end{align*}
         where the second inequality follows by removing a term from
the last parenthesis and rearranging, and the third from the parallelogram inequality $\|\vec{a}\|^2 + \|\vec{b}\|^2 \ge \frac{1}{2}\|\vec{a}+\vec{b}\|^2$ valid for all choices of vectors $\vec{a},\vec{b}$ and norm $\|\cdot\|$.

        In order to recognize \circled{A} on the left-hand side, we add the quantity $\sum_{t=1}^T \langle\vec{\ell}^t,\vec{w}^t-\xhat\rangle$ on both sides, and obtain
        \begin{align*}
            \circled{A} &\le \frac{1}{\eta}(\regu(\vec{w}^T) - \regu(\vec{w}^{0})) + \sum_{t=1}^T \Big(\langle \vec{\ell}^t, \vec{w}^t - \xhat\rangle + \langle\vec{L}^{t-1}, \vec{w}^t - \vec{w}^{t-1}\rangle\Big) -\frac{1}{4\eta} \sum_{t=1}^{T-1}\|\vec{x}^{t+1} - \vec{x}^{t}\|^2\\
                &= \frac{1}{\eta}(\regu(\vec{w}^T) - \regu(\vec{w}^{0})) + \sum_{t=1}^T \Big(\langle \vec{L}^t, \vec{w}^t \rangle - \langle\vec{L}^{t-1}, \vec{w}^{t-1}\rangle - \langle \vec{\ell}^t, \xhat\rangle\Big) - \frac{1}{4\eta}\sum_{t=1}^{T-1}\|\vec{x}^{t+1} - \vec{x}^{t}\|^2\\
                &= \frac{1}{\eta}(\regu(\vec{w}^T) - \regu(\vec{w}^{0})) + \langle \vec{L}^T, \vec{w}^T - \xhat \rangle - \frac{1}{4\eta}\sum_{t=1}^{T-1}\|\vec{x}^{t+1} - \vec{x}^{t}\|^2,\numberthis\label{eq:ftrl last step}
        \end{align*}
        where we simplified the telescopic sum $\sum_{t=1}^T\langle\vec{L}^t, \vec{w}^t\rangle - \langle\vec{L}^{t-1},\vec{w}^{t-1}\rangle = \langle\vec{L}^T, \vec{w}^T\rangle$ in the last step. Finally, using \cref{eq:crucial} with $\vec{g} = \vec{L}^T, \vec{\xi} = \xhat$, we can write
        \[
            \frac{1}{\eta}\regu(\xhat) + \langle\vec{L}^T, \xhat\rangle \ge \frac{1}{\eta}\regu(\vec{w}^T) + \langle\vec{L}^T,\vec{w}^T\rangle \implies
            \frac{1}{\eta}\regu(\vec{w}^T) + \langle\vec{L}^T, \vec{w}^T - \xhat\rangle \le \frac{1}{\eta}\regu(\xhat),
        \]
        and substituting the last expression into~\eqref{eq:ftrl last step}, we obtain
        \begin{equation}\label{eq:ftrl part A}
            \circled{A} \le \frac{1}{\eta}(\regu(\xhat) - \regu(\vec{w}^0)) - \sum_{t=1}^{T-1}\frac{1}{4\eta}\|\vec{x}^{t+1} - \vec{x}^{t}\|^2 \le \frac{\regu(\xhat)}{\eta} - \frac{1}{4\eta}\sum_{t=1}^{T-1}\|\vec{x}^{t+1} - \vec{x}^{t}\|^2.
        \end{equation}
        \item[\circled{B}] By applying the generalized Cauchy-Schwarz inequality and \cref{lem:psi lipschitz},
            \[
                \langle \vec{\ell}^t - \vec{m}^t, \vec{x}^t - \vec{w}^t \rangle \le \|\vec{\ell}^t - \vec{m}^t\|_*\,\|\vec{x}^t - \vec{w}^t\| \le \eta \| \vec{\ell}^t - \vec{m}^t\|_*^2.
            \]
            Hence,
            \begin{equation}\label{eq:ftrl part B}
                \circled{B} = \sum_{t=1}^T \langle \vec{\ell}^t - \vec{m}^t, \vec{x}^t - \vec{w}^t \rangle \le \eta \sum_{t=1}^T \| \vec{\ell}^t - \vec{m}^t\|_*^2.
            \end{equation}
    \end{itemize}
    Finally, summing the bounds for \circled{A}~\eqref{eq:ftrl part A} and for \circled{B}~\eqref{eq:ftrl part B}, we obtain the statement.
\end{proof}

\section{Analysis of (Predictive) OMD}\label{app:proofs omd}

In the proof of \cref{prop:omd bound} we will use the two following technical lemmas.

\begin{lemma}\label{lem:amgm}
    For any $\vec{a}, \vec{b} \in\bbR^n$ and $\rho > 0$, it holds that
    $\displaystyle
        \langle \vec{a}, \vec{b} \rangle \le \frac{\rho}{2}  \|\vec{a}\|_*^2 + \frac{1}{2\rho} \|\vec{b}\|^2
    $.
\end{lemma}
\begin{proof}
    By the arithmetic mean-geometric mean inequality, we have
    \[
        \frac{\rho}{2} \|\vec{a}\|_*^2 + \frac{1}{2\rho}\|\vec{b}\|^2 = \frac{1}{2}\mleft(\rho \|\vec{a}\|_*^2 + \frac{1}{\rho}\|\vec{b}\|^2\mright) \ge \sqrt{\|\vec{a}\|_*^2 \cdot\|\vec{b}\|^2} = \|\vec{a}\|_* \cdot\|\vec{b}\| \ge \langle \vec{a}, \vec{b}\rangle,
    \]
    where we used the generalized Cauchy-Schwarz inequality in the last step.
\end{proof}

\begin{lemma}\label{lem:divergence triangle}
    Let $\cD \subseteq \bbR^d$ be closed and convex, let $\vec{g} \in\bbR^n, \vec{c} \in \cD$, and let $\regu : \cD \to \bbR_{\ge 0}$ be a $1$-strongly convex differentiable regularizer with respect to some norm $\|\cdot\|$, and let $\|\cdot\|_*$ be the dual norm to $\|\cdot\|$. Then,
    \[
        \vec{a}^* \defeq \argmin_{\hat{\vec{a}}\in\cD} \mleft\{\langle \vec{g}, \hat{\vec{a}}\rangle + \frac{1}{\eta} \div{\hat{\vec{a}}}{\vec{c}}\mright\}
    \]
    is well defined (that is, the minimizer exists and is unique), and for all $\hat{\vec{a}} \in \cD$ satisfies the inequality
    \[
        \langle \vec{g}, \vec{a}^* - \hat{\vec{a}}\rangle \le \frac{1}{\eta}\Big(\div{\hat{\vec{a}}}{\vec{c}} - \div{\hat{\vec{a}}}{\vec{a}^*} - \div{\vec{a}^*}{\vec{c}}\Big).
    \]
\end{lemma}
\begin{proof}
    The necessary first-order optimality conditions for the argmin problem in the statement is
    \[
        \mleft\langle \nabla_{\vec{a}}\mleft[\langle \vec{g}, {\vec{a}}\rangle + \frac{1}{\eta} \div{{\vec{a}}}{\vec{c}}\mright](\vec{a}^*), \hat{\vec{a}} - \vec{a}^* \mright\rangle \ge 0 \quad \forall \,\hat{\vec{a}}\in \cD.
    \]
    Expanding the gradient, we have that for all $\hat{\vec{a}}\in\cD$
    \[
        \mleft\langle \vec{g} + \frac{1}{\eta}\Big(\nabla \regu(\vec{a}^*) - \nabla\regu(\vec{c})\Big), \hat{\vec{a}} - \vec{a}^* \mright\rangle \ge 0
        \iff
        \langle \vec{g}, \vec{a}^* - \hat{\vec{a}} \rangle \le \frac{1}{\eta}\Big\langle
            \nabla \regu(\vec{a}^*) - \nabla\regu(\vec{c}), \hat{\vec{a}} - \vec{a}^* 
        \Big\rangle.
    \]
    Finally, noting that
    \begin{align*}
        \Big\langle\nabla \regu(\vec{a}^*) - \nabla\regu(\vec{c}), \hat{\vec{a}} - \vec{a}^*\Big\rangle &= \Big(\regu(\hat{\vec{a}}) - \regu(\vec{c}) - \langle \nabla \regu(\vec{c}), \hat{\vec{a}} - \vec{c} \rangle \Big)\\
            &\hspace{1cm} -\Big(\regu(\hat{\vec{a}}) - \regu(\vec{a}^*) - \langle \nabla \regu(\vec{a}^*), \hat{\vec{a}} - \vec{a}^* \rangle \Big)\\
            &\hspace{1cm} -\Big(\regu(\vec{a}^*) - \regu(\vec{c}) - \langle \nabla \regu(\vec{c}), \vec{a}^* - \vec{c} \rangle \Big)\\
        &= \div{\hat{\vec{a}}}{\vec{c}} - \div{\hat{\vec{a}}}{\vec{a}^*} - \div{\vec{a}^*}{\vec{c}},
    \end{align*}
    yields the statement.
\end{proof}

\begin{restatable}{proposition}{propomdbound}\label{prop:omd bound}
    Let $\regu : \cD \to \bbR_{\ge 0}$ be a $1$-strongly convex differentiable regularizer with respect to some norm $\|\cdot\|$, and let $\|\cdot\|_*$ be the dual norm to $\|\cdot\|$. For all $\xhat \in \cD$, all $\eta > 0$, and all times $T$, the regret cumulated by (predictive) OMD (\cref{algo:predictive omd}) compared to any fixed strategy $\hat{\vec{x}} \in \cD$ is bounded as
    \begin{equation}\label{eq:oomd regret bound}
        R^T(\xhat) \le \frac{\div{\xhat}{\vec{z}^0}}{\eta} + \eta\sum_{t=1}^T \|\vec{\ell}^t - \vec{m}^t\|_*^2 - \frac{1}{8\eta} \sum_{t=1}^{T-1} \|\vec{x}^{t+1} - \vec{x}^t\|^2.
    \end{equation}
\end{restatable}
\begin{proof}
    We combine several techniques and insights from the original works of \citet{Rakhlin13:Online} and \citet{Syrgkanis15:Fast}.
    For any $\xhat \in \cD$,
    \[
        R^T(\xhat) = \sum_{t=1}^T \langle \vec{\ell}^t, \vec{x}^t - \xhat\rangle = \sum_{t=1}^T \bigg(%
            \underbrace{\langle \vec{\ell}^t - \vec{m}^t, \vec{x}^t - \vec{z}^t \rangle}_{\circled{A}}%
            +%
            \underbrace{\langle \vec{m}^t, \vec{x}^t - \vec{z}^t \rangle}_{\circled{B}}%
            +%
            \underbrace{\langle \vec{\ell}^t, \vec{z}^t - \xhat \rangle}_{\circled{C}}%
        \bigg)
    \]
    We now bound each of the three terms on the right-hand side:
    \begin{itemize}[leftmargin=12mm,nolistsep,itemsep=1mm]
        \item[\circled{A}] We use \cref{lem:amgm} with $\rho = 2\eta$ to bound the first term:
            \[
                \langle \vec{\ell}^t - \vec{m}^t, \vec{x}^t - \vec{z}^{t} \rangle \le \eta\|\vec{\ell}^t - \vec{m}^t\|_*^2 + \frac{1}{4\eta}\|\vec{x}^t - \vec{z}^{t}\|^2.
            \]
        \item[\circled{B}\,\circled{C}] In order to bound these terms, we use \cref{lem:divergence triangle}:
            \begin{align*}
                \langle \vec{m}^t, \vec{x}^t - \vec{z}^{t}\rangle &\le \frac{1}{\eta}\Big( \div{\vec{z}^{t}}{\vec{z}^{t-1}} - \div{\vec{z}^t}{\vec{x}^t} - \div{\vec{x}^t}{\vec{z}^{t-1}} \Big)\\
                \langle \vec{\ell}^t, \vec{z}^t - \xhat\rangle &\le \frac{1}{\eta} \Big( \div{\xhat}{\vec{z}^{t-1}} - \div{\xhat}{\vec{z}^t} - \div{\vec{z}^t}{\vec{z}^{t-1}} \Big)
            \end{align*}
    \end{itemize}
    Hence, combining all bounds, we have that for any $\xhat \in \cD$,
    \begin{align*}
        R^T(\xhat) &\le \sum_{t=1}^T \bigg(%
            \eta\|\vec{\ell}^t - \vec{m}^t\|_*^2 + \frac{1}{4\eta}\|\vec{x}^t - \vec{z}^{t}\|^2%
            \\%
            &\hspace{3cm}+\frac{1}{\eta}\Big(  \div{\xhat}{\vec{z}^{t-1}} - \div{\xhat}{\vec{z}^t} - \div{\vec{z}^t}{\vec{x}^t} - \div{\vec{x}^t}{\vec{z}^{t-1}} \Big)%
        \bigg)\\
            &\le \sum_{t=1}^T \bigg(%
            \eta\|\vec{\ell}^t - \vec{m}^t\|_*^2 + \frac{1}{4\eta}\|\vec{x}^t - \vec{z}^{t}\|^2 + \frac{1}{\eta}\Big(\div{\xhat}{\vec{z}^{t-1}} - \div{\xhat}{\vec{z}^t}\Big)\\%
            &\hspace{7cm}- \frac{1}{2\eta}\Big(\|\vec{x}^t - \vec{z}^t\|^2 + \|\vec{x}^t - \vec{z}^{t-1}\|^2\Big)\bigg)\\
            &= \sum_{t=1}^T \bigg(%
            \eta\|\vec{\ell}^t - \vec{m}^t\|_*^2 - \frac{1}{4\eta}\|\vec{x}^t - \vec{z}^{t}\|^2 - \frac{1}{2\eta}\|\vec{x}^t - \vec{z}^{t-1}\|^2 + \frac{1}{\eta}\Big(\div{\xhat}{\vec{z}^{t-1}} - \div{\xhat}{\vec{z}^t}\Big)\bigg)\\
            &\le \sum_{t=1}^T \bigg(%
            \eta\|\vec{\ell}^t - \vec{m}^t\|_*^2 - \frac{1}{4\eta}\|\vec{x}^t - \vec{z}^{t}\|^2 - \frac{1}{4\eta}\|\vec{x}^t - \vec{z}^{t-1}\|^2 + \frac{1}{\eta}\Big(\div{\xhat}{\vec{z}^{t-1}} - \div{\xhat}{\vec{z}^t}\Big)\bigg)\\
    \end{align*}
    where we used the fact that $\div{\vec{a}}{\vec{b}} \ge \frac{1}{2}\|\vec{a} - \vec{b}\|^2$ for all $\vec{a},\vec{b}\in\cD$ (because $\regu$ is $1$-strongly convex by hypothesis) in the second inequality. Since the differences of divergences on the right-hand side are telescopic, we further obtain
    \begin{align*}
        R^T(\xhat) &\le \frac{\div{\xhat}{\vec{z}^0} - \div{\xhat}{\vec{z}^t}}{\eta}%
           +\eta\sum_{t=1}^T\|\vec{\ell}^t - \vec{m}^t\|_*^2%
           -\frac{1}{4\eta}\sum_{t=1}^T\|\vec{x}^t - \vec{z}^{t}\|^2%
           -\frac{1}{4\eta}\sum_{t=1}^T\|\vec{x}^t - \vec{z}^{t-1}\|^2\\
        &\le \frac{\div{\xhat}{\vec{z}^0}}{\eta}%
            +\eta\sum_{t=1}^T\|\vec{\ell}^t - \vec{m}^t\|_*^2%
            -\frac{1}{4\eta}\sum_{t=1}^T\|\vec{x}^t - \vec{z}^{t}\|^2%
            -\frac{1}{4\eta}\sum_{t=1}^T\|\vec{x}^t - \vec{z}^{t-1}\|^2\\
        &= \frac{\div{\xhat}{\vec{z}^0}}{\eta}%
            +\eta\sum_{t=1}^T\|\vec{\ell}^t - \vec{m}^t\|_*^2%
            -\frac{1}{4\eta}\sum_{t=1}^T\|\vec{x}^t - \vec{z}^{t}\|^2%
            -\frac{1}{4\eta}\sum_{t=0}^{T-1}\|\vec{x}^{t+1} - \vec{z}^t\|^2\\
        &\le \frac{\div{\xhat}{\vec{z}^0}}{\eta}%
            +\eta\sum_{t=1}^T\|\vec{\ell}^t - \vec{m}^t\|_*^2%
            -\frac{1}{4\eta}\sum_{t=1}^{T-1}\|\vec{x}^t - \vec{z}^{t}\|^2%
            -\frac{1}{4\eta}\sum_{t=1}^{T-1}\|\vec{x}^{t+1} - \vec{z}^t\|^2\\
        &= \frac{\div{\xhat}{\vec{z}^0}}{\eta}%
            +\eta\sum_{t=1}^T\|\vec{\ell}^t - \vec{m}^t\|_*^2%
            -\frac{1}{4\eta}\sum_{t=1}^{T-1}\Big(\|\vec{x}^t - \vec{z}^{t}\|^2 + \|\vec{x}^{t+1} - \vec{z}^t\|^2\Big),
    \end{align*}
    where we used the nonnegativity of divergences in the second inequality, and some trivial manipulation of summation indices in the later steps. Finally, we use the triangle inequality for the norm $\|\cdot\|$ to conclude that at all $t=1,\dots, T-1$
    \[
        \|\vec{x}^t - \vec{z}^{t}\|^2 + \|\vec{x}^{t+1} - \vec{z}^t\|^2 \ge \frac{1}{2}\|\vec{x}^{t+1} - \vec{x}^t\|^2,
    \]
    and hence for all $\xhat \in \cD$
    \[
        R^T(\xhat) \le \frac{\div{\xhat}{\vec{z}^0}}{\eta}%
            +\eta\sum_{t=1}^T\|\vec{\ell}^t - \vec{m}^t\|_*^2%
            -\frac{1}{8\eta}\sum_{t=1}^{T-1} \|\vec{x}^{t+1} - \vec{x}^t\|^2.
    \]
\end{proof}

    When $\nabla\regu(\vec{z}^0) = \vec{0}$ as in \cref{line:omd setup} in \cref{algo:predictive omd}, $\div{\xhat}{\vec{z}^0} \le \regu(\xhat)$ and so \cref{prop:omd bound} becomes
    \begin{corollary}
          For all $\xhat \in \cD$, all $\eta > 0$, and all times $T$, the regret cumulated by (predictive) OMD (\cref{algo:predictive omd}) compared to any fixed strategy $\hat{\vec{x}} \in \cD$ is bounded as
    \begin{equation}\label{eq:oomd regret bound}
        R^T(\xhat) \le \frac{\regu(\xhat)}{\eta} + \eta\sum_{t=1}^T \|\vec{\ell}^t - \vec{m}^t\|_*^2 - \frac{1}{8\eta} \sum_{t=1}^{T-1} \|\vec{x}^{t+1} - \vec{x}^t\|^2.
    \end{equation}
    \end{corollary} 

%% file: text/appendix_olo_to_approachability.tex
\section{Online Linear Optimization to Approachability}\label{app:proofs olo to approachability}

\propolotoapproachability*
\begin{proof}
    Let $\cK \defeq C^\circ \cap \bbB^n_2$. As proved by \citet{Abernethy11:Blackwell}, the distance from the generic point $\vec{z}$ to the convex cone $C$ can be computed as
    \[
        \min_{\hat{\vec{s}}\in C} \|\hat{\vec{s}} - \vec{z}\|_2 = \max_{\thetahat \in \cK}\, \langle \thetahat, \vec{z} \rangle.
    \]
    Hence,
    \begin{align}
        \min_{\hat{\vec{s}}\in C}\ \mleft\| \hat{\vec{s}} - \frac{1}{T}\sum_{t=1}^T \vec{u}(\vec{x}^t, \vec{y}^t)\mright\|_2 &= \max_{\thetahat \in \cK}\,\mleft\langle \thetahat, \frac{1}{T}\sum_{t=1}^T \vec{u}(\vec{x}^t, \vec{y}^t)\mright\rangle\nonumber\\
            &= - \frac{1}{T}\sum_{t=1}^T \langle\vec{\theta}^t, \vec{\ell}^t\rangle + \frac{1}{T}\max_{\thetahat \in \cK} \mleft\{ \sum_{t=1}^T \langle\vec{\ell}^t, \vec{\theta}^t - \thetahat \rangle\mright\}\\
            &= - \frac{1}{T}\sum_{t=1}^T \langle\vec{\theta}^t, \vec{\ell}^t\rangle + \frac{1}{T}\max_{\thetahat \in \cK} R(\hat{\vec{\theta}})\label{eq:distance bound}
    \end{align}
    where the second step uses $\vec{\ell}^t = -\vec{u}(\vec{x}^t,\vec{y}^t)$.
    Since $\vec{\theta}^t \in \cD \subseteq C^\circ$, the halfspace $H^t \defeq \{\vec{z} : \langle \vec{\theta}^t, \vec{z}\rangle \le 0\}$ contains $C$ at all times $t$. Furthermore, by construction $\vec{x}^t$ forces $H^t$, and so $\langle \vec{\theta}^t, \vec{\ell}^t\rangle = -\langle \vec{\theta}^t, \vec{u}(\vec{x}^t, \vec{y}^t)\rangle \ge 0$, and therefore
    \begin{equation}\label{eq:bound sum}
        -\frac{1}{T} \sum_{t=1}^T \langle \vec{\theta}^t, \vec{\ell}^t\rangle \le 0.
    \end{equation}
    Plugging~\eqref{eq:bound sum} into~\eqref{eq:distance bound} yields the statement.
\end{proof}

%% file: text/appendix_rm_rmplus.tex
\section{Connections between FTRL, OMD and RM, \rmp{}}\label{app:proofs rm rmp}

\lemsimplextoblackwell*
\begin{proof}
    The regret $R^T(\xhat)$ cumulated by PRM and \prmp{} satisfies
        \begin{align*}
            \frac{1}{T} R^T(\xhat) &= \frac{1}{T} \sum_{t=1}^T \Big(\langle \vec{\ell}^t, \vec{x}^t\rangle - \langle\vec{\ell}^t, \xhat\rangle\Big) = \sum_{t=1}^T \Big(\langle \vec{\ell}^t, \vec{x}^t\rangle\langle \vec{1}, \xhat\rangle - \langle\vec{\ell}^t, \xhat\rangle\Big)\\
                &= \mleft\langle \frac{1}{T}\sum_{t=1}^T \langle \vec{\ell}^t,\vec{x}^t\rangle\vec{1} - \vec{\ell}^t, \xhat\mright\rangle
                = \mleft\langle \frac{1}{T}\sum_{t=1}^T \vec{u}(\vec{x}^t, \vec{\ell}^t), \xhat\mright\rangle\\
                &= \min_{\hat{\vec{s}} \in \bbR^n_{\le 0}} \mleft\langle -\hat{\vec{s}} + \frac{1}{T}\sum_{t=1}^T \vec{u}(\vec{x}^t, \vec{\ell}^t), \xhat\mright\rangle,\numberthis\label{eq:regret analysis}
        \end{align*}
        where we used the fact that $\xhat \in \Delta^{\!n}$ in the second equality, and the fact that $\min_{\hat{\vec{s}}\in\bbR^n_{\le 0}} \langle -\hat{\vec{s}}, \xhat\rangle = 0$ since $\xhat \ge \vec{0}$. Applying the Cauchy-Schwarz inequality to the right-hand side of~\eqref{eq:regret analysis}, we obtain
        \begin{align*}
            \frac{1}{T} R^T(\xhat) &\le \min_{\hat{\vec{s}} \in \bbR^n_{\le 0}} \mleft\| -\hat{\vec{s}} + \frac{1}{T}\sum_{t=1}^T \vec{u}(\vec{x}^t, \vec{\ell}^t)\mright\|_2 \|\xhat\|_2.
        \end{align*}
        So, using the fact that $\|\xhat\|_2 \le 1$ for any $\xhat \in \Delta^{\!n}$
        \begin{align*}
            \frac{1}{T} R^T(\xhat) &\le \min_{\hat{\vec{s}} \in \bbR^n_{\le 0}} \mleft\| -\hat{\vec{s}} + \frac{1}{T}\sum_{t=1}^T \vec{u}(\vec{x}^t, \vec{\ell}^t)\mright\|_2
        \end{align*}
        as we wanted to show.
\end{proof}

\thmrmisftrl*
\begin{proof}
    Given the definition of $\Gamma$ and \cref{algo:olo to approachability}, at all times $t$, $\cL_\eta^\text{ftrl}$ observes loss $-\vec{u}(\vec{x}^t,\vec{\ell}^t)$, where $\vec{u}(\vec{x}^t, \vec{\ell}^t) \defeq \langle \vec{\ell}^t, \vec{x}^t\rangle\vec{1} - \vec{\ell}^t$ is the vector-valued payoff in $\Gamma$ and measures the increase of regret at time $t$ relative to each vertex of the simplex. For the specific choice of domain $\cD = \bbR_{\ge 0}^n$ and regularizer $\regu(\vec{x}) = \frac{1}{2}\|\vec{x}\|_2^2$, the computation of the next iterate (\cref{line:ftrl next strategy} in non-predictive FTRL, \cref{algo:predictive ftrl}) reduces to
    \begin{align*}
      \vec{\theta}^t &= \argmin_{\hat{\vec{x}} \in \bbR^n_{\ge 0}} \mleft\{\mleft\langle-\sum_{t=1}^T \vec{u}(\vec{x}^t,\vec{\ell}^t), \xhat\mright\rangle + \frac{1}{2\eta}\|\xhat\|_2^2\mright\} \\
        &= \argmin_{\hat{\vec{x}} \in \bbR^n_{\ge 0}} \mleft\{\mleft\langle-2\eta\sum_{t=1}^T \vec{u}(\vec{x}^t,\vec{\ell}^t), \xhat\mright\rangle + \|\xhat\|_2^2\mright\} \\
        &= \argmin_{\hat{\vec{x}} \in \bbR^n_{\ge 0}} \mleft\|\xhat - \eta \sum_{t=1}^T \vec{u}(\vec{x}^t, \vec{\ell}^t)\mright\|_2^2
        = \mleft[\eta \sum_{t=1}^T \vec{u}(\vec{x}^t, \vec{\ell}^t)\mright]^+
        = \eta \mleft[\sum_{t=1}^T \vec{u}(\vec{x}^t, \vec{\ell}^t)\mright]^+.
    \end{align*}
    Now, the value of $\eta > 0$ does not affect the forcing action that needs to be played on Line 3 of \cref{algo:olo to approachability}. Indeed, whenever $\vec{\theta}^t \neq 0$, $\vec{g}(\vec{\theta}^t) = \vec{\theta}^t / \|\vec{\theta}^t\|_1$, so $\eta$ cancels out in the fraction and at all $t$,
    \[
        \vec{x}^t = \frac{\mleft[\sum_{t=1}^T \vec{u}(\vec{x}^t, \vec{\ell}^t)\mright]^+}{\mleft\|\mleft[\sum_{t=1}^T \vec{u}(\vec{x}^t, \vec{\ell}^t)\mright]^+\mright\|_1}.
    \]
    This is exactly the strategy output by RM.
\end{proof}

\thmrmpisomd*
\begin{proof}
      Given the definition of $\Gamma$ and \cref{algo:olo to approachability}, at all times $t$, $\cL_\eta^\text{omd}$ observes loss $-\vec{u}(\vec{x}^t,\vec{\ell}^t)$, where $\vec{u}(\vec{x}^t, \vec{\ell}^t) \defeq \langle \vec{\ell}^t, \vec{x}^t\rangle\vec{1} - \vec{\ell}^t$ is the vector-valued payoff in $\Gamma$ and measures the increase of regret at time $t$ relative to each vertex of the simplex.  In the non-predictive version of OMD $\vec{m}^t = \vec{0}$, \cref{line:omd next xt} in \cref{algo:predictive omd} is equivalent to $\argmin \div{\xhat}{\vec{z}^{t-1}} = \vec{z}^{t-1}$. Hence, for the specific choice of domain $\cD = \bbR_{\ge 0}^n$ and regularizer $\regu(\vec{x}) = \frac{1}{2}\|\vec{x}\|_2^2$, the computation of the next iterate (\cref{line:omd next zt} in non-predictive OMD, \cref{algo:predictive omd}) reduces to
    \begin{align*}
        \vec{\theta}^t =\ \vec{z}^{t-1} &= \argmin_{\hat{\vec{z}} \in \bbR^n_{\ge 0}} \mleft\{\Big\langle-\vec{u}(\vec{x}^{t-1},\vec{\ell}^{t-1}), \zhat\Big\rangle + \frac{1}{\eta}\div{\zhat}{\vec{z}^{t-2}}\mright\} \\
        &= \argmin_{\hat{\vec{z}} \in \bbR^n_{\ge 0}} \mleft\{\Big\langle-\vec{u}(\vec{x}^{t-1},\vec{\ell}^{t-1}), \zhat\Big\rangle + \frac{1}{2\eta}\|\zhat - \vec{z}^{t-2}\|_2^2\mright\} \\
        &= \argmin_{\hat{\vec{z}} \in \bbR^n_{\ge 0}} \Big\|\zhat - \vec{z}^{t-2} - \eta\, \vec{u}(\vec{x}^{t-1}, \vec{\ell}^{t-1})\Big\|_2^2
        = \mleft[\vec{z}^{t-2} + \eta\, \vec{u}(\vec{x}^{t-1}, \vec{\ell}^{t-1})\mright]^+ \\ &= \mleft[\vec{\theta}^{t-1} + \eta\, \vec{u}(\vec{x}^{t-1}, \vec{\ell}^{t-1})\mright]^+.\numberthis\label{eq:omd proj}
    \end{align*}
    Since $\vec{\theta}^1 = \vec{z}^0 = \vec{0}$, the only effect of the step size $\eta$ is a rescaling of all iterates $\{\vec{\theta}^t\}$ by a constant. However, the forcing action $\vec{g}(\vec{\theta}^t) = \vec{\theta}^t / \|\vec{\theta}^t\|_1$ is invariant to positive rescaling of $\vec{\theta}^t$. For this reason, all choices of $\eta > 0$ result in the same iterates being output by the algorithm. So, in particular we can assume without loss of generality that $\eta = 1$ in~\eqref{eq:omd proj}, which corresponds exactly to the update step in \rmp{}.
\end{proof}

%% file: text/appendix_predictive_approachability.tex
\section{Predictive Blackwell Approachability and Predictive RM, \rmp{}}\label{app:prm prmp}

\proppredictiveblackwell*
\begin{proof}
  As shown by \citet{Abernethy11:Blackwell}, a Blackwell approachability game
  with a non-conic target set can be converted to a conic target set at the
  cost of a factor 2 in the distance bound. Hence, we assume that $S$ is a
  closed convex cone, and use the construction of \cref{algo:olo to approachability}
  instantiated with the FTRL algorithm with domain $\cD = S^\circ$, regularizer
  $\regu(\vec{x}) = \frac{1}{2}\|\vec{x}\|_2^2$, and step size parameter $\eta > 0$.
  \cref{prop:olo to approachability}, along with the aforementioned factor $2$
  reduction from generic convex target set to conic target set, implies that
    \begin{align*}
        \min_{\hat{\vec{s}}\in C}\ \mleft\| \hat{\vec{s}} - \frac{1}{T}\sum_{t=1}^T \vec{u}(\vec{x}^t, \vec{y}^t)\mright\|_2 &\le \frac{2}{T}\max_{\xhat \in S^\circ \cap \bbB_2^n} R^T(\xhat)\\
        &\le \frac{2}{T}\max_{\xhat \in S^\circ \cap \bbB_2^n} \mleft(\frac{\|\xhat\|_2^2}{2\eta} + \eta \sum_{t=1}^T \|\vec{u}(\vec{x}^t, \vec{y}^t) - \vec{v}^t\|_2^2\mright)\\
        &\le \frac{2}{T}\mleft(\frac{1}{2\eta} + \eta \sum_{t=1}^T \|\vec{u}(\vec{x}^t, \vec{y}^t) - \vec{v}^t\|_2^2\mright)
    \end{align*}
  where the second inequality follows from expanding the regret bound for FTRL
  (\cref{prop:ftrl bound}), and the third inequality follows from the fact that
  $\xhat \in \bbB_2^n$. Setting $\eta = \frac{1}{\sqrt{T}}$ yields the result.
\end{proof}

\thmprmprmp*
\begin{proof}
    Given the definition of $\Gamma$ and \cref{algo:olo to approachability}, at all times $t$, $\cL_\eta^\text{ftrl*}$ and $\cL_\eta^\text{omd*}$ observe loss $-\vec{u}(\vec{x}^t,\vec{\ell}^t)$, where $\vec{u}(\vec{x}^t, \vec{\ell}^t) \defeq \langle \vec{\ell}^t, \vec{x}^t\rangle\vec{1} - \vec{\ell}^t$ is the vector-valued payoff in $\Gamma$ and measures the increase of regret at time $t$ relative to each vertex of the simplex. Furthermore, at all $t$ the prediction given to $\cL_\eta^\text{ftrl*}$ and $\cL_\eta^\text{omd*}$ is $-\vec{v}^t$ (Line 2, \cref{algo:olo to approachability}). We now break up the analysis according to the OLO oracle used.

    \paragraph{$\cL_\eta^\text{ftrl*}$ corresponds to Predictive RM} For the specific choice of domain $\cD = \bbR^n_{\ge 0}$ and regularizer $\regu = \|\cdot\|_2^2$, \cref{line:ftrl next strategy} in \cref{algo:predictive ftrl} has the closed-form solution
    \[
        \vec{\theta}^{t} = \mleft[-\eta \mleft(-\sum_{t=1}^T \vec{u}(\vec{x}^t, \vec{\ell}^t) - \vec{v}^t\mright)\mright]^+ = \eta\mleft[\sum_{t=1}^T \vec{u}(\vec{x}^t, \vec{\ell}^t) + \vec{v}^t\mright]^+.
    \]
    Since the forcing action $\vec{g}(\vec{\theta}^t) = \vec{\theta^t} / \|\vec{\theta^t}\|_1$ is invariant to positive constants, we see that the action $\vec{x}^t$ picked by \cref{algo:olo to approachability} (Line 3) is the same for all values of $\eta > 0$ and is computed as
    \begin{equation}\label{eq:prm xt}
        \vec{x}^t  = \frac{\mleft[\sum_{t=1}^T \vec{u}(\vec{x}^t, \vec{\ell}^t) + \vec{v}^t\mright]^+}{\mleft\|\mleft[\sum_{t=1}^T \vec{u}(\vec{x}^t, \vec{\ell}^t) + \vec{v}^t\mright]^+\mright\|_1}.
    \end{equation}
    provided $\vec{\theta}^t \neq \vec{0}$, and is an arbitrary vector $\vec{x}^t \in \Delta^n$ otherwise, in accordance with the analysis of the approachability of halfspaces in $\Gamma$ (\cref{sec:ftrl omd rm rmp}). By using the definition of $\vec{u}(\vec{x}^t, \vec{\ell}^t) \defeq \langle \vec{\ell}^t, \vec{x}^t\rangle\vec{1} - \vec{\ell}^t$ and $\vec{v}^t \defeq \langle \vec{m}^t, \vec{x}^{t-1}\rangle\vec{1} - \vec{m}^t$, we see that at all times $t$ the iterates produced by Line 4 in \cref{algo:prm} are exactly as in~\eqref{eq:prm xt}.

\paragraph{$\cL_\eta^\text{omd*}$ corresponds to Predictive \rmp{}} For the specific choice of domain $\cD = \bbR^n_{\ge 0}$ and regularizer $\regu = \|\cdot\|_2^2$, as already note in the proof of \cref{thm:rmp is omd}, \cref{line:omd next zt} in Predictive OMD (\cref{algo:predictive omd}) has the closed-form solution
    \begin{align}
      \vec{z}^{t} = \mleft[\vec{z}^{t-1} + \eta\,\vec{u}(\vec{x}^t, \vec{\ell}^t)\mright]^+\label{eq:prmp1}
    \end{align}
    at all $t$. Similarly, \cref{line:omd next xt} in Predictive OMD (\cref{algo:predictive omd}) has the closed-form solution
    \begin{align}
        \vec{\theta}^t &= \mleft[\vec{z}^{t-1} + \eta \vec{v}^t\mright]^+.\label{eq:prmp2}
    \end{align}
    Since both~\eqref{eq:prmp1} and~\eqref{eq:prmp2} are homogeneous in $\eta > 0$ (that is, the only effect of $\eta$ is to rescale all $\vec{\theta}^t$ and $\vec{z}^t$ by the same constant) and the forcing action $\vec{g}(\vec{\theta}^t) = \vec{\theta}^t / \|\vec{\theta}^t\|_1$ for $\Gamma$ is invariant to positive rescaling of $\vec{\theta}^t$, we see that \cref{algo:olo to approachability} outputs the same iterates no matter the choice of step size parameter $\eta > 0$. In particular, we can assume without loss of generality that $\eta = 1$. In that case, \cref{eq:prmp1} corresponds exactly to Line 7 in \prmp{} (\cref{algo:prmp}), and line \cref{eq:prmp2} corresponds exactly to Line 4.

    \paragraph{Regret analysis} The regret $R^T(\xhat)$ cumulated by PRM and \prmp{} satisfies
    \begin{align*}
        \frac{1}{T} R^T(\xhat) &= \frac{1}{T} \sum_{t=1}^T \Big(\langle \vec{\ell}^t, \vec{x}^t\rangle - \langle\vec{\ell}^t, \xhat\rangle\Big) = \sum_{t=1}^T \Big(\langle \vec{\ell}^t, \vec{x}^t\rangle\langle \vec{1}, \xhat\rangle - \langle\vec{\ell}^t, \xhat\rangle\Big)\\
            &= \mleft\langle \frac{1}{T}\sum_{t=1}^T \langle \vec{\ell}^t,\vec{x}^t\rangle\vec{1} - \vec{\ell}^t, \xhat\mright\rangle
            = \mleft\langle \frac{1}{T}\sum_{t=1}^T \vec{u}(\vec{x}^t, \vec{\ell}^t), \xhat\mright\rangle\\
            &= \min_{\hat{\vec{s}} \in \bbR^n_{\le 0}} \mleft\langle -\hat{\vec{s}} + \frac{1}{T}\sum_{t=1}^T \vec{u}(\vec{x}^t, \vec{\ell}^t), \xhat\mright\rangle,\numberthis\label{eq:regret analysis}
    \end{align*}
    where we used the fact that $\xhat \in \Delta^{\!n}$ in the second equality, and the fact that $\min_{\hat{\vec{s}}\in\bbR^n_{\le 0}} \langle -\hat{\vec{s}}, \xhat\rangle = 0$ since $\xhat \ge \vec{0}$. Applying the Cauchy-Schwarz inequality to the right-hand side of~\eqref{eq:regret analysis}, we obtain
    \begin{align*}
        \frac{1}{T} R^T(\xhat) &\le \min_{\hat{\vec{s}} \in \bbR^n_{\le 0}} \mleft\| -\hat{\vec{s}} + \frac{1}{T}\sum_{t=1}^T \vec{u}(\vec{x}^t, \vec{\ell}^t)\mright\|_2 \|\xhat\|_2.
    \end{align*}
    So, using the fact that $\|\xhat\|_2 \le 1$ for any $\xhat \in \Delta^{\!n}$, and applying \cref{prop:olo to approachability},
    \begin{align}
        \frac{1}{T} R^T(\xhat) &\le \min_{\hat{\vec{s}} \in \bbR^n_{\le 0}} \mleft\| -\hat{\vec{s}} + \frac{1}{T}\sum_{t=1}^T \vec{u}(\vec{x}^t, \vec{\ell}^t)\mright\|_2 \le \frac{1}{T} \max_{\xhat' \in \bbR^n_{\ge 0} \cap \bbB_2^n}{R_\cL^T(\xhat')},\label{eq:prmp regret step 2}
    \end{align}
    where $R_\cL^T$ is the regret cumulated by the OLO oracle used in \cref{algo:olo to approachability}---in our case, $\cL_\eta^\text{ftrl*}$ for PRM and $\cL_\eta^\text{omd*}$ for \prmp{}.
    In either case ($\cL = \cL_\eta^\text{ftrl*}$ or $\cL = \cL_\eta^\text{omd*}$),
        \cref{prop:oco bound} offers a bound on $R_\cL^T(\xhat)$ that holds for
        all $\xhat \in\cD = \bbR_{\ge 0}^n$. So, in particular the bound holds
        for all points in $\cK = \bbR_{\ge 0}^n \cap \bbB_2^n \subseteq \cD$.
    Consequently,
    \begin{equation}\label{eq:prmp regret step 3}
        \max_{\xhat' \in \bbR^n_{\ge 0} \cap \bbB_2^n}{R_\cL^T(\xhat')} \le \max_{\xhat' \in \bbR^n_{\ge 0} \cap \bbB_2^n} \mleft\{\frac{\|\xhat'\|_2^2}{2\eta} + \eta\sum_{t=1}^T \|\vec{u}(\vec{x}^t, \vec{\ell}^t) - \vec{v}^t\|_2^2\mright\} \le \frac{1}{2\eta} + \eta\sum_{t=1}^T \|\vec{u}(\vec{x}^t, \vec{\ell}^t) - \vec{v}^t\|_2^2,
    \end{equation}
    where we used the fact that $\xhat'\in\bbB_2^n$ in the last step. Substituting~\eqref{eq:prmp regret step 3} into~\eqref{eq:prmp regret step 2}, we have
    \[
        R^T(\xhat) \le \frac{1}{2\eta} + \eta\sum_{t=1}^T \|\vec{u}(\vec{x}^t, \vec{\ell}^t) - \vec{v}^t\|_2^2.
    \]
    Since we have shown above that the iterates produced by the algorithm are independent of $\eta > 0$, we can minimize the right-hand side over $\eta > 0$, obtaining the bound
    \[
        R^T(\xhat) \le \sqrt{2}\mleft(\sum_{t=1}^T \|\vec{u}(\vec{x}^t, \vec{\ell}^t) - \vec{v}^t\|_2^2\mright)^{\!\!1/2}.
    \]
    Finally, expanding the definition of $\vec{u}(\vec{x}^t, \vec{\ell}^t) \defeq \langle \vec{\ell}^t, \vec{x}^t\rangle\vec{1} - \vec{\ell}^t$ and $\vec{v}^t \defeq \langle \vec{m}^t, \vec{x}^{t-1}\rangle\vec{1} - \vec{m}^t$, we obtain the statement.
\end{proof}

%% file: text/appendix_cfr.tex
\section{Extensive-Form Games and Counterfactual Regret Minimization}\label{app:efg cfr}

\newcommand{\regret}[1]{R^T_{#1}}
\newcommand{\sympl}[1]{\Delta^{\!#1}}
\newcommand{\sprm}[3]{$(#1,#2,#3)$-stable-predictive}
\newcommand{\seqf}[1]{X^{\triangle}_{#1}}
\newcommand{\sfrm}[1]{\cR^{\triangle}_{#1}}
\newcommand{\sfrT}[1]{R^{\triangle,T}_{#1}}
\newcommand{\sfm}[2]{\vec{m}^{\triangle,#1}_{#2}}
\newcommand{\sfell}[2]{\vec{\ell}^{\triangle,#1}_{#2}}
\newcommand{\sfx}[2]{\vec{x}^{\triangle,#1}_{#2}}
\newcommand{\localrm}[1]{\hat\cR_{#1}}
\newcommand{\nextv}[1]{\mathcal{C}_{#1}}
\newcommand{\subt}[1]{{\triangle}_{#1}}
\newcommand{\laminarregret}[2]{{\hat R}^{#2}_{#1}}

An extensive-form game is a game played on a game tree. Each player in an extensive-form game faces a sequential decision process. A sequential decision process is a tree consisting of two types of nodes: \emph{decision nodes} and \emph{observation nodes}. We denote the set of decision nodes as $\cJ$, and the set of observation nodes with $\cK$. At each decision node $j \in \cJ$, the agent picks an action according to a distribution $\vec{x}_j \in \Delta^{n_j}$ over the set $A_j$ of $n_j = |A_j|$ actions available at that decision node, and the process moves to the observation node that is reached by following the edge corresponding to the selected action at $j$, if any.
At each observation point $k \in \cK$, the agent receives one out of $n_k$ possible signals; the set of signals that the agent can observe is denoted as $S_k$. After the signal is received, the process moves to the decision node that is reached by following the edge corresponding to the signal at $k$.

The observation node that is reached by the agent after picking action $a \in A_j$ at decision point $j\in \cJ$ is denoted by $\rho(j, a)$. Likewise, the decision node reached by the agent after observing signal $s\in S_k$ at observation point $k \in \cK$ is denoted by $\rho(k, s)$. The set of all observation points reachable from $j \in \cJ$ is denoted as $\nextv{j} \defeq \{\rho(j, a): a\in A_j\}$. Similarly, the set of all decision points reachable from $k \in \cK$ is denoted as $\nextv{k} \defeq \{\rho(k, s): s\in S_k\}$. To ease the notation, sometimes we will use the notation $\nextv{ja}$ to mean $\nextv{\rho(j,a)}$.

Pairs $z = (j, a)$ with $j\in\cJ, a\in A_j$ for which $\rho(j,a) = \emptyset$ are called \emph{terminal sequences} and have an associated payoff vector $(u(z), -u(z))$ (that is, we assume the game is zero sum). We denote the set of all terminal sequences (also called \emph{leaves}) with $Z$.

\paragraph{Sequence Form for Sequential Decision Processes}
Given a strategy $\{\vec{x}_j\}_{j\in\cJ}$ for the player, its sequence-form representation~\citep{Stengel96:Efficient}, denoted $\mu(\vec{x})$ is defined as the vector indexed over $\{(j,a): j\in\cJ, a\in A_j\}$ whose entry corresponding to a generic pair $(j,a)$ is the product of the probability of all actions on the path from the root of the decision process to $(j,a)$. We denote the range of $\mu$, that is the set of all possible sequence-form strategies as the $\vec{x}_j$ vary arbitrarily over $\Delta^{|A_j|}$ as $Q$. We call $Q$ the sequence-form strategy space of the player.

It is well-known that a Nash equilibrium in a two-player zero-sum extensive form game can be expressed as a bilinear saddle point problem
\[
    \min_{\vec{q}_1 \in Q_1}\max_{\vec{q}_2 \in Q_2} \vec{q}_1^{\!\top} \mat{A} \vec{q}_2,
\]
where $Q_1$ and $Q_2$ are the sequence-form strategy spaces of Player 1 and 2, respectively, and $\mat{A}$ is a suitable game-dependent matrix. It is also common knowledge that by letting regret minimizers for $Q_1$ and $Q_2$ play against each other, we can sole the bilinear saddle point above (e.g., \citet{Farina19:Online}). So, we now focus on the task of constructing a regret minimizer for a sequence-form strategy space.

\subsection{Counterfactual Regret Minimization}
The counterfactual regret minimization framework \citep{Zinkevich07:Regret} provides a way of constructing a regret minimization for the sequence-form strategy space of a player by combining independent regret minimizers \emph{local} to each of the player's decision points $j \in \cJ$. At each $j\in \cJ$, the corresponding regret minimizer---denoted $\cR_j$---is responsible for selecting the strategy $\vec{x}_j^t$ at all times $t$.

CFR achieves its goal by setting the losses observed by the local regret minimizers in a specific way. In particular, let $\vec{\ell}^t$ be the loss at time $t$ relative to the whole sequence-form strategy space $Q$ of the player. Then, for each decision point $j \in \cJ$, the regret minimizer $\cR_j$ local at $j$ is fed the loss vector $\vec{\ell}^{t}_j \in \bbR^{|A_j|}$, whose entries are defined as
\[
    \vec{\ell}^{t}_j[a] \defeq \vec{\ell}^t[(j,a)] + \sum_{j' \in \cC_{ja}} V^t_{j'}\numberthis\label{eq:counterfactual loss}
\]
for each $a \in A_j$, where
\[
    V^t_j \defeq \sum_{a \in A_j} \vec{x}_j^t[a] \mleft(\vec{\ell}^t[(j,a)] + \sum_{j' \in \cC_{ja}} V^t_{j'}\mright)\qquad \forall j \in \cJ.\numberthis\label{eq:Vt}
\]

\begin{theorem}[Laminar regret decomposition, \citep{Farina19:Online}]\label{thm:laminar}
 At all times $T$, the regret $R^T$ cumulated by the CFR algorithm can be bounded as
    \[
    \max_{\xhat \in Q} R^T(\xhat) \le \max_{\xhat \in Q}\sum_{j\in\cJ} \xhat[\sigma(j)] \cdot R^T_j(\xhat_j)
    \]
    where $R^T_j$ denotes the regret cumulated by the local regret minimizer $\cR_j$ at decision point $j$.
\end{theorem}

\cref{thm:laminar} in particular implies that if all local regret minimizers $\cR_j$ ($j \in \cJ$) guarantee $O(T^{1/2})$ regret, then so does the overall algorithm, that is $R^T(\xhat) = O(T^{1/2})$ for all $\xhat \in Q$. 

\subsection{Counterfactual Loss Predictions}

We now describe the construction of the counterfactual loss predictions, starting from a generic prediction $\vec{m}^t$ for $\vec{\ell}^t$ relative to the whole sequence-form strategy space $Q$ of the player. In order to maintain symmetry with \cref{eq:counterfactual loss} and \cref{eq:Vt}, for each decision point $j \in \cJ$, the regret minimizer $\cR_j$ local at $j$ is fed the loss prediction vector $\vec{m}^{t}_j \in \bbR^{|A_j|}$, whose entries are defined as
\[
    \vec{m}^{t}_j[a] \defeq \vec{m}^t[(j,a)] + \sum_{j' \in \cC_{ja}} W^t_{j'}
\]
for each $a \in A_j$, where
\[
    W^t_j \defeq \sum_{a \in A_j} \vec{x}_j^t[a] \mleft(\vec{m}^t[(j,a)] + \sum_{j' \in \cC_{ja}} W^t_{j'}\mright)\qquad \forall j \in \cJ.
\]

It important to observe that the counterfactual loss prediction $\vec{m}^t_j$
depends on the decisions produced at time $t$ in the subtree rooted at $j$. In other words, in order to construct the prediction for what loss $\cR_j$
will observe after producing the decision $\vec{x}_j^t$, we use the ``future'' decisions from the subtrees under $j$. 

In our experiments, we always set $\vec{m}^t = \vec{\ell}^{t-1}$. This is a common choice, that in other algorithms (not ours) is known to lead to asymptotically lower regret than $O(T^{1/2})$~\citep{Syrgkanis15:Fast,Farina19:Optimistic,Farina19:Optimistic}. 

%% file: text/appendix_experiments.tex
\section{Description of the Game Instances}\label{app:games}

\begin{description}
  \item[Kuhn poker] (\cref{game:kuhn,game:kuhn13})  is a standard benchmark in the EFG-solving community \citep{Kuhn50:Simplified}. In Kuhn poker, each player puts an ante worth $1$ into the pot. Each player is then privately dealt one card from a deck that contains $R$ unique cards. Then, a single round of betting then occurs, with the following dynamics. First, Player $1$ decides to either check or bet $1$. Then,
  \begin{itemize}[nolistsep]
  \item If Player 1 checks Player 2 can check or raise $1$.
    \begin{itemize}[nolistsep]
      \item If Player 2 checks a showdown occurs; if Player 2 raises Player 1 can fold or call.
        \begin{itemize}
          \item If Player 1 folds Player 2 takes the pot; if Player 1 calls a showdown occurs.
          \end{itemize}
        \end{itemize}
      \item If Player 1 raises Player 2 can fold or call.
        \begin{itemize}[nolistsep]
        \item If Player 2 folds Player 1 takes the pot; if Player 2 calls a showdown occurs.
        \end{itemize}
      \end{itemize}
      When a showdown occurs, the player with the higher card wins the pot and the game immediately ends.

    We used $R=3$ in \cref{game:kuhn} (this corresponds to the original game as introduced by \citet{Kuhn50:Simplified}), while in \cref{game:kuhn13} we used $R=13$.

    \item[Leduc poker] (\cref{game:leduc3,game:leduc5,game:leduc9,game:leduc13}) is another standard benchmark in the EFG-solving
community~\cite{Southey05:Bayes}. The game is played with a deck of $R$ unique
cards, each of which appears exactly twice in the deck. The game is composed of two rounds. In the
first round, each player places an ante of $1$ in the pot and is dealt a single private card. A round of betting then takes place, with Player 1 acting first. At most two bets are allowed per player. Then, a card is is revealed face up and another
round of betting takes place, with the same dynamics described above. After the two betting round, if one of the players has a pair with the public card, that player
wins the pot. Otherwise, the player with the higher card wins the pot. All bets in the first
round are worth $1$, while all bets in the second round are $2$.

We set $R=3$ in \cref{game:leduc3}, $R=5$ in \cref{game:leduc5}, $R=9$ in \cref{game:leduc9}, and $R=13$ in \cref{game:leduc13}.

\item[Small matrix] (\cref{game:sm}) is a small $2 \times 2$ matrix game. Given a mixed strategy $\vec{x} = ( x_1, x_2) \in \Delta^2$ for Player 1 and a mixed strategy $\vec{y} = (y_1, y_2)\in\Delta^2$ for Player 2, the payoff function for player 1 is defined as
      \[
        u(\vec{x}, \vec{y}) \defeq 5 x_1 y_1 - x_1 y_2 + x_2 y_2.
      \]
    This game was found by \cite{Farina19:Optimistic} to be a hard instance for the \cfrp{} game.

    \item[Goofspiel] (\cref{game:goof4,game:goof5}) This is another popular benchmark game, originally proposed by \citet{Ross71:Goofspiel}. It is a two-player card game, employing three identical decks
of $k$ cards each whose values range from $1$ to $k$. At the beginning of the game, each player gets dealt a full deck as their hand, and the third deck (the ``prize'' deck) is shuffled and put face down on the board. In each turn, the topmost card from the prize deck is revealed. Then, each player privately picks a card from their hand. This card acts as a bid to win the card that was just revealed from the prize deck. The selected cards are simultaneously revealed, and the highest one
wins the prize card. If the players' played cards are equal, the prize card is split.  The players’ score are computed as the sum of the values of the prize cards they have won.
        In \cref{game:goof4} the value of $k$ is $k = 4$, while in \cref{game:goof5} $k=5$.

    \item[Limited-information Goofspiel] (\cref{game:goof4_li,game:goof5_li}) This is a variant of the Goofspiel game used by~\citet{Lanctot09:Monte}. In this variant, in each turn the players do not reveal their cards. Rather, they show their cards to a fair umpire, which determines which player has played the highest card and should therefore received the prize card. In case of tie, the umpire directs the players to discard the prize card just like in the Goofspiel game.
        In \cref{game:goof4_li} the number of cards in each deck is $k = 4$, while in \cref{game:goof5_li} $k=5$.

    \item[Pursuit-evasion] (\cref{game:search4,game:search5,game:search6}) is a security-inspired pursuit-evasion game played on the graph shown in Figure~\ref{fig:search_game}. It is a zero-sum variant of the one used by \citet{Kroer18:Robust}, and a similar search game  has been considered by \citet{Bosansky14:Exact} and \citet{Bosansky15:Sequence}.

\begin{figure}[ht]
  \centering
  \begin{tikzpicture}
        \path[fill=black!10!white] (.8, -1.4) rectangle (1.6, 1.4);
        \path[fill=black!10!white] (3.2, -1.4) rectangle (4.0, 1.4);
        \node at (1.2, -1.7) {$P_1$};
        \node at (3.6, -1.7) {$P_2$};

          \node[draw, circle, minimum width=.6cm, inner sep=0] (A) at (0, 0) {$S$};
          \node[draw, circle, minimum width=.6cm] (B) at (1.2, 1) {};
          \node[draw, circle, minimum width=.6cm] (C) at (1.2, 0) {};
          \node[draw, circle, minimum width=.6cm] (D) at (1.2, -1) {};
            \node[draw, circle, minimum width=.6cm] (E) at (2.4, 1) {};
            \node[draw, circle, minimum width=.6cm] (F) at (2.4, 0) {};
            \node[draw, circle, minimum width=.6cm] (G) at (2.4, -1) {};
              \node[draw, circle, minimum width=.6cm] (H) at (3.6, 1) {};
              \node[draw, circle, minimum width=.6cm] (I) at (3.6, 0) {};
              \node[draw, circle, minimum width=.6cm] (J) at (3.6, -1) {};
                    \node[draw, circle, minimum width=.6cm,inner sep=0] (K) at (4.8, 1) {$5$};
                    \node[draw, circle, minimum width=.6cm,inner sep=0] (L) at (4.8, 0) {$10$};
                    \node[draw, circle, minimum width=.6cm,inner sep=0] (M) at (4.8, -1) {$3$};

        \draw[thick,->] (A) edge (B);
        \draw[thick,->] (A) edge (C);
        \draw[thick,->] (A) edge (D);
        \draw[thick,->] (B) edge (E);
        \draw[thick,->] (C) edge (F);
        \draw[thick,->] (D) edge (G);
        \draw[thick,->] (E) edge (F);
        \draw[thick,->] (G) edge (F);
        \draw[thick,->] (E) edge (H);
        \draw[thick,->] (F) edge (I);
        \draw[thick,->] (G) edge (J);
        \draw[thick,->] (H) edge (K);
        \draw[thick,->] (I) edge (L);
        \draw[thick,->] (J) edge (M);

        \draw[thick,gray,dashed] (B) edge (C);
        \draw[thick,gray,dashed] (C) edge (D);

        \draw[thick,gray,dashed] (H) edge (I);
        \draw[thick,gray,dashed] (I) edge (J);
    \end{tikzpicture}
  \caption{The graph on which the search game is played.}
  \label{fig:search_game}
\end{figure}
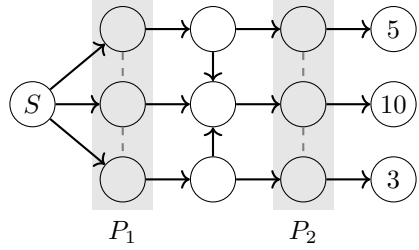

In each turn, the attacker and the defender act simultaneously. The defender controls two patrols, one per each
respective patrol areas labeled $P_1$ and $P_2$. Each patrol can move by one step along the grey dashed lines, or stay in place. The attacker starts from the leftmost node (labeled $S$) and at each turn can move to
any node adjacent to its current position by following the black directed edges. The attacker can also choose to wait
in place for a time step in order to hide all their traces. If a patrol visits a
node that was previously visited by the attacker, and the attacker did not wait
to clean up their traces, they will see that the attacker was there. The goal of the attacker is to reach any of the rightmost nodes, whose corresponding payoffs are $5$, $10$, or $3$, respectively, as indicated in \cref{fig:search_game}. If at any time the attacker and any patrol
meet at the same node, the attacker is loses the game, which leads to
a payoff of $-1$ for the attacker and of $1$ for the defender. The game times out after $m$ simultaneous moves, in which case both players defender
receive payoffs $0$. In \cref{game:search4} we set $m=4$, in \cref{game:search5} we set $m=5$ and in \cref{game:search6} we set $m=6$.

    \item[Battleship] (\cref{game:bs23_3turns,game:bs23_4turns}) is a parametric version of a classic board game, where two
competing fleets take turns shooting at each other~\citep{Farina19:Correlation}.
At the beginning of the game, the players take turns at secretly placing a set
of ships on separate grids (one for each player) of size $3\times 2$. Each ship
has size 2 (measured in terms of contiguous grid cells) and a value of $4$, and
must be placed so that all the cells that make up the ship are fully contained
within each player’s grids and do not overlap with any other ship that the
player has already positioned on the grid. After all ships have been placed. the
players take turns at firing at their opponent. Ships that have been hit at all
their cells are considered sunk. The game continues until either one player has
sunk all of the opponent’s ships, or each player has completed $R$ shots. At the
end of the game, each player’s payoff is calculated as the sum of the values of
the opponent’s ships that were sunk, minus the sum of the values of ships which
that player has lost.

In \cref{game:bs23_3turns} we set $R = 3$, while in \cref{game:bs23_4turns} we set $R=4$.

\item[River Endgame] (\cref{game:river_endgame7}) The river endgame is structured and
  parameterized as follows. The game is parameterized by the conditional
  distribution over hands for each player, current pot size, board state ($5$
  cards dealt to the board), and a betting abstraction. First, Chance deals out
  hands to the two players according to the conditional hand distribution. Then,
  Libratus has the choice of folding, checking, or betting by a number of
  multipliers of the pot size: 0.25x, 0.5x, 1x, 2x, 4x, 8x, and all-in. If
  Libratus checks and the other player bets then Libratus has the choice of
  folding, calling (i.e. matching the bet and ending the betting), or raising by
  pot multipliers 0.4x, 0.7x, 1.1x, 2x, and all-in. If Libratus bets and the
  other player raises Libratus can fold, call, or raise by 0.4x, 0.7x, 2x, and
  all-in. Finally when facing subsequent raises Libratus can fold, call, or
  raise by 0.7x and all-in. When faced with an initial check, the opponent can
  fold, check, or raise by 0.5x, 0.75x, 1x, and all-in. When faced with an
  initial bet the opponent can fold, call, or raise by 0.7x, 1.1x, and all-in.
  When faced with subsequent raises the opponent can fold, call, or raise by
  0.7x and all-in. The game ends whenever a player folds (the other player wins
  all money in the pot), calls (a showdown occurs), or both players check as
  their first action of the game (a showdown occurs). In a showdown the player
  with the better hands wins the pot. The pot is split in case of a tie.
  The specific endgame we use is subgame 4 from the set of open-sourced Libratus subgames at
  \url{https://github.com/Sandholm-Lab/LibratusEndgames}.

  \item[Liar's dice] (\cref{game:ld_new}) is another standard benchmark in the
      EFG-solving community~\citep{Lisy15:Online}. In our instantiation, each
      of the two players initially privately rolls an unbiased $6$-face die.
      The first player begins bidding, announcing any face value up to $6$ and
      the minimum number of dice that the player believes are showing that value
      among the dice of both players. Then, each player has two choices during
      their turn: to make a higher bid, or to challenge the previous bid by
      declaring the previous bidder a ``liar''. A bid is higher than the
      previous one if either the face value is higher, or the number of dice is
      higher. If the current player challenges the previous bid, all dice are
      revealed. If the bid is valid, the last bidder wins and obtains a reward
      of $+1$ while the challenger obtains a negative payoff of $-1$. Otherwise,
      the challenger wins and gets reward $+1$, and the last bidder obtains reward of $-1$.
\end{description}

\section{Additional Experimental Results}

\begin{figure}[H]
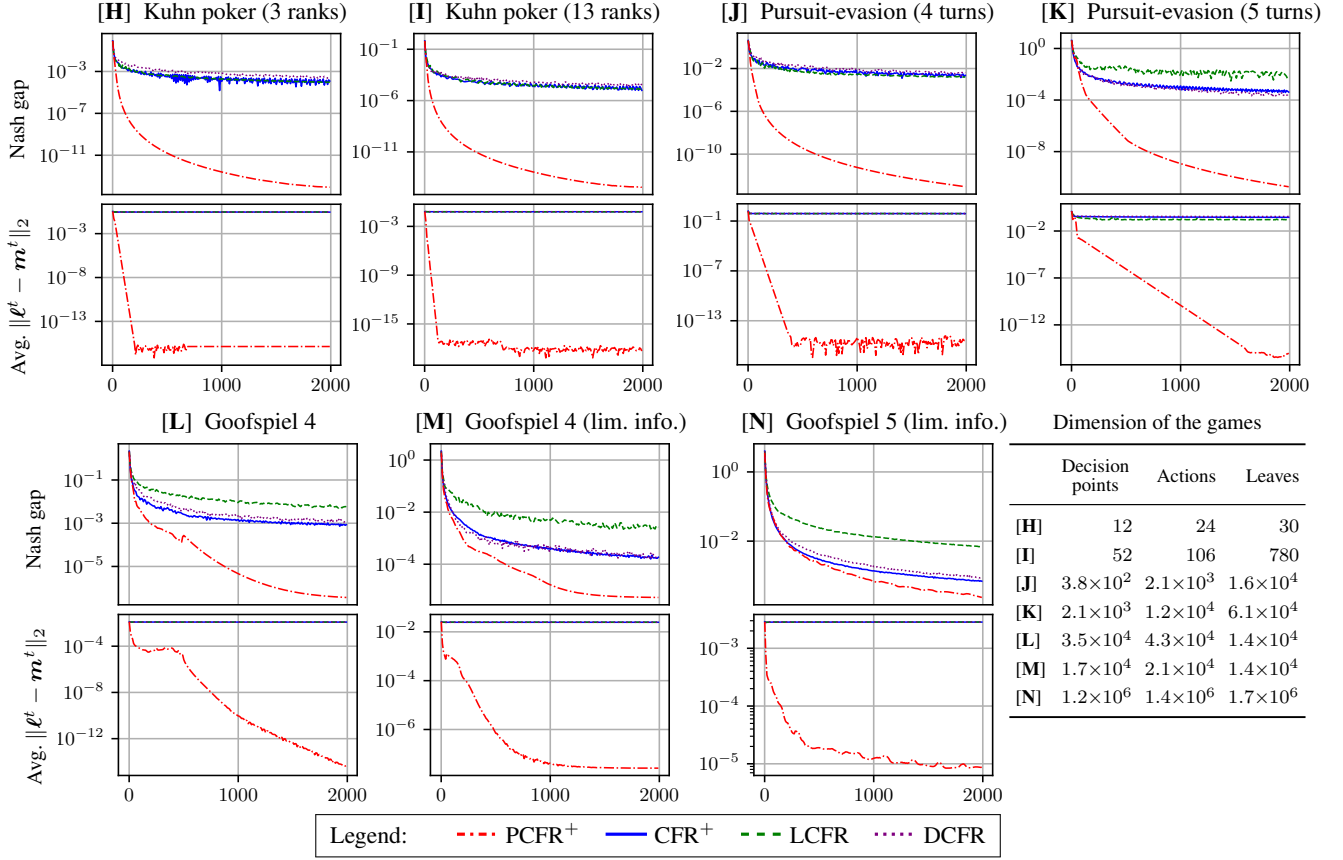
\centering%
    \addplotA{kuhn}{Kuhn poker (3 ranks)}%
    \hspace{-1mm}\addplotA*{kuhn13}{Kuhn poker (13 ranks)}%
    \hspace{-1mm}\addplotA*{search4}{Pursuit-evasion (4 turns)}%
    \hspace{-1mm}\addplotA*{search5}{Pursuit-evasion (5 turns)}\\
    \addplotA{goof4}{Goofspiel 4}%
    \hspace{-1mm}\addplotA*{goof4_li}{Goofspiel 4 (lim. info.)}%
    \hspace{-1mm}\addplotA*{goof5_li}{Goofspiel 5 (lim. info.)}%
    \hspace{-1mm}\scalebox{.9}{\begin{minipage}[t]{4.4cm}%
      \small%
      \centering Dimension of the games\\[1.08mm]
      {%
        \fontsize{8}{8}\selectfont%
        \setlength{\tabcolsep}{.9mm}%
        \renewcommand{\arraystretch}{1.5}%
        \sisetup{
            scientific-notation=true,
            round-precision=1,
            round-mode=places,
            exponent-product={\mkern-4mu\times\mkern-4.5mu}
        }%
      \begin{tabular}{l>{\raggedleft\let\newline\\\arraybackslash\hspace{0pt}}m{1.1cm}rr}
        \toprule
         & \centering Decision points & Actions & Leaves \\%
        \midrule%
        \ref*{game:kuhn}       &        $12$ &         $24$ &         $30$ \\
        \ref*{game:kuhn13}     &        $52$ &        $106$ &        $780$ \\
        \ref*{game:search4}    &   \num{382} &   \num{2081} &  \num{15898} \\
        \ref*{game:search5}    &  \num{2078} &  \num{11899} &  \num{61084} \\
        \ref*{game:goof4}      & \num{34952} &  \num{42658} &  \num{13824} \\
        \ref*{game:goof4_li}   & \num{17432} &  \num{21298} &  \num{13824} \\
        \ref*{game:goof5_li} & \num{1175330} & \num{1428452} & \num{1728000} \\
        \bottomrule
      \end{tabular}
      }
    \end{minipage}}\\
    \centering\makelegendA{}
    \caption{Performance of \pcfrp, \cfrp, DCFR, and LCFR on EFGs. In all plots, the x axis is the number of iterations of each algorithm. For each game, the top plot shows that the Nash gap on the y axis (on a log scale), the bottom plot shows and the average prediction error (on a log scale).}
\end{figure}

\begin{figure}[t]
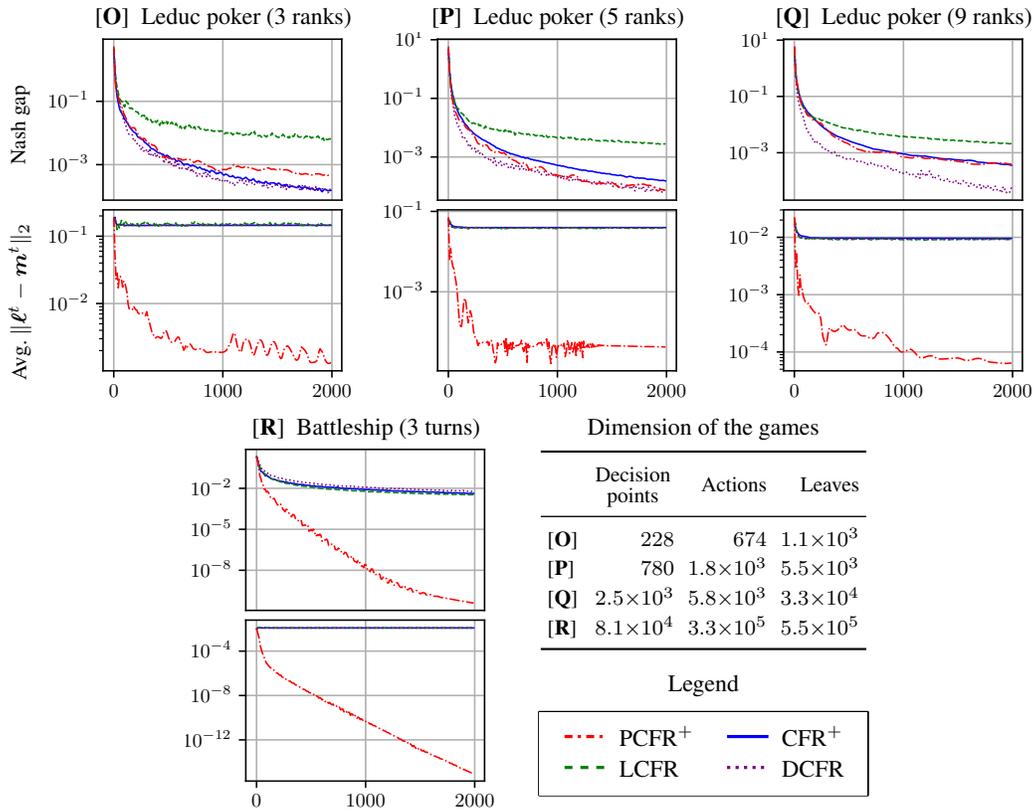
\centering%
    \addplotA{leduc3}{Leduc poker (3 ranks)}%
    \hspace{2mm}\addplotA*{leduc5}{Leduc poker (5 ranks)}%
    \hspace{2mm}\addplotA*{leduc9}{Leduc poker (9 ranks)}%
    \hspace{2mm}\addplotA*{bs23_3turns}{Battleship (3 turns)}%
    \hspace{4mm}\begin{minipage}[t]{4.4cm}%
      \small%
      \centering Dimension of the games\\[1.08mm]
      {%
        \fontsize{8}{8}\selectfont%
        \setlength{\tabcolsep}{.9mm}%
        \renewcommand{\arraystretch}{1.4}%
        \sisetup{
            scientific-notation=true,
            round-precision=1,
            round-mode=places,
            exponent-product={\mkern-4mu\times\mkern-4.5mu}
        }%
      \begin{tabular}{l>{\raggedleft\let\newline\\\arraybackslash\hspace{0pt}}m{1.1cm}rr}
        \toprule
         & \centering Decision points & Actions & Leaves \\%
        \midrule%
        \ref*{game:leduc3}     &       $228$ &        $674$ &   \num{1116}\\
        \ref*{game:leduc5}     &       $780$ &   \num{1822} &   \num{5500}\\
        \ref*{game:leduc9}     &  \num{2484} &   \num{5798} &  \num{32724}\\
        \ref*{game:bs23_3turns}& \num{81027} & \num{327070} & \num{552132} \\
        \bottomrule
      \end{tabular}
      }\\[.3cm]
      \centering Legend\\[.5mm]
      \begin{tcolorbox}[
        boxsep=0pt,
        left=1pt,right=1pt,top=5pt,bottom=4pt,
        boxrule=.5pt,
        colback=black!0!white,
        colframe=black,
        arc=0pt
      ]\centering
        \begin{tabular}{ll}
         \linesty{pred,dash dot}~~\pcfrp{}&
         \linesty{pblue}~~\cfrp{}\\[.6mm]
         \linesty{pgreen,dashed}~~LCFR&
         \linesty{ppurple,dotted}~~DCFR\\
        \end{tabular}
      \end{tcolorbox}
    \end{minipage}\\
    \caption{Performance of \pcfrp, \cfrp, DCFR, and LCFR on EFGs. In all plots, the x axis is the number of iterations of each algorithm. For each game, the top plot shows that the Nash gap on the y axis (on a log scale), the bottom plot shows and the average prediction error (on a log scale).}
\end{figure}

In all games but Leduc 13 (\cref{game:leduc13}), \pcfrp{} significantly outperforms all other algorithms, by 2-8 orders of magnitude. In Leduc 13, \pcfrp{} outperforms \cfrp{} but not the DCFR algorithm. \cfrp{} is equivalent or slightly superior to DCFR, except in Leduc 13, where it outperforms \cfrp{} by slightly less of one order of magnitude. This is in line with the experimental results presented in the body of this paper, where we found that DCFR performs significantly better than \cfrp{} in poker games but not other domains.

\cfrp{}, LCFR, and DCFR perform similarly in the Small matrix game (\cref{game:sm}), and in particular all exhibit slower than $T^{-1}$ convergence. This is not the case for our predictive algorithm \pcfrp{}. This confirms that Small matrix is a hard instance for non-predictive methods but not for predictive methods, as already observed by~\citet{Farina19:Optimistic}.

In all game instances, we empirically find that the prediction error decreases quickly to extremely small values. This suggests that \pcfrp{} might enjoy stability guarantees similar to predictive FTRL and OMD~\citep{Syrgkanis15:Fast}. Exploring such properties is an interesting future research direction.

\paragraph{Correlation between game structure and \pcfrp{} performance} The empirical investigation of \pcfrp{} shows that in most classes of games \pcfrp{} performs significantly better than \cfrp{} and DCFR, while in other games (such as the poker games and Liar's Dice) predictivity seems to be less useful or even detrimental. It is natural to wonder what game structures can benefit from the use of predictive methods and what do not. While we do not currently have a good answer to that question, we have collected here some thoughts and observations.

\begin{itemize}[nolistsep,itemsep=1mm]
    \item \emph{Size}. Some predictive methods proposed in the past were found to only produce a speedup in small games, and perform worse than the state of the art in large games~\citep{Farina19:Optimistic}. This is {not} the case for \pcfrp{}: the river endgame and Liar's Dice are not the largest games in our dataset. So, size does not seem to be a good predictor for whether predictive \cfrp{} is beneficial over \cfrp{} and DCFR.

    \item \emph{Number of terminal states}. The river endgame and Liar's Dice both have a large ratio between the number of terminal states (leaves) and number of decision points. On the other hand, the pursuit-evasion game with 5 turns (\cref{game:search5}) has a significantly larger ratio than Liar's Dice but unlike in Liar's Dice, predictivity yields a speedup of more than 6 orders of magnitude on the Nash gap.

    \item \emph{Private information}. Poker games and Liar's Dice have a strong private information structure: a chance node distributes independent private initial states for the two players, and each player has no information about the opponent's state. This is in contrast with, for example, the Battleship games, where each player is \emph{not} handed a random configuration for their ships by the chance player, but rather privately picks one configuration. This shows that the ``amount of private information'' alone is not a good discriminator for when predictivity can be useful.
    
    \item \emph{Private information induced by chance nodes}. From the discussion in the previous bullet, we conjecture that the way the private information arises (for example, through "dealing out cards" like in Poker games or "rolling a die" as in Liar's Dice) might affect whether predictivity helps or hurts convergence to Nash equilibrium. We leave pursuing this direction open. It is not immediately clear how one could formalize that metric.
\end{itemize}

\subsection{Comparison between Linear and Quadratic Averaging in \pcfrp{} and \cfrp{}}

We also investigated the performance of \cfrp{} with quadratic averaging in all games, as well as the performance of \pcfrp{} with linear averaging. The experimental results are shown in \cref{fig:comparison cfrp quadratic avg 1,fig:comparison cfrp quadratic avg 2}.
Since only the averaging that is used when computing the (approximate) Nash equilibrium varies, but not the iterates themselves, the prediction errors are independent of the averaging variant used. Therefore, in the prediction error plots we only report one curve for each of the two algorithms.

\cfrp{} with quadratic averaging of iterates performs similarly to \cfrp{} with linear averaging. \pcfrp{} with linear averaging performs similarly or slightly better than \pcfrp{} with quadratic averaging in two games. It performs better than \cfrp{} with either linear or quadratic averaging in 11 games, and worse than both in two games (no-limit Texas hold'em river endgame and Leduc poker). We conclude that the speedup of \pcfrp{} is mostly due to the use of loss predictions, rather than the particular averaging of iterates.

\begin{figure}[H]
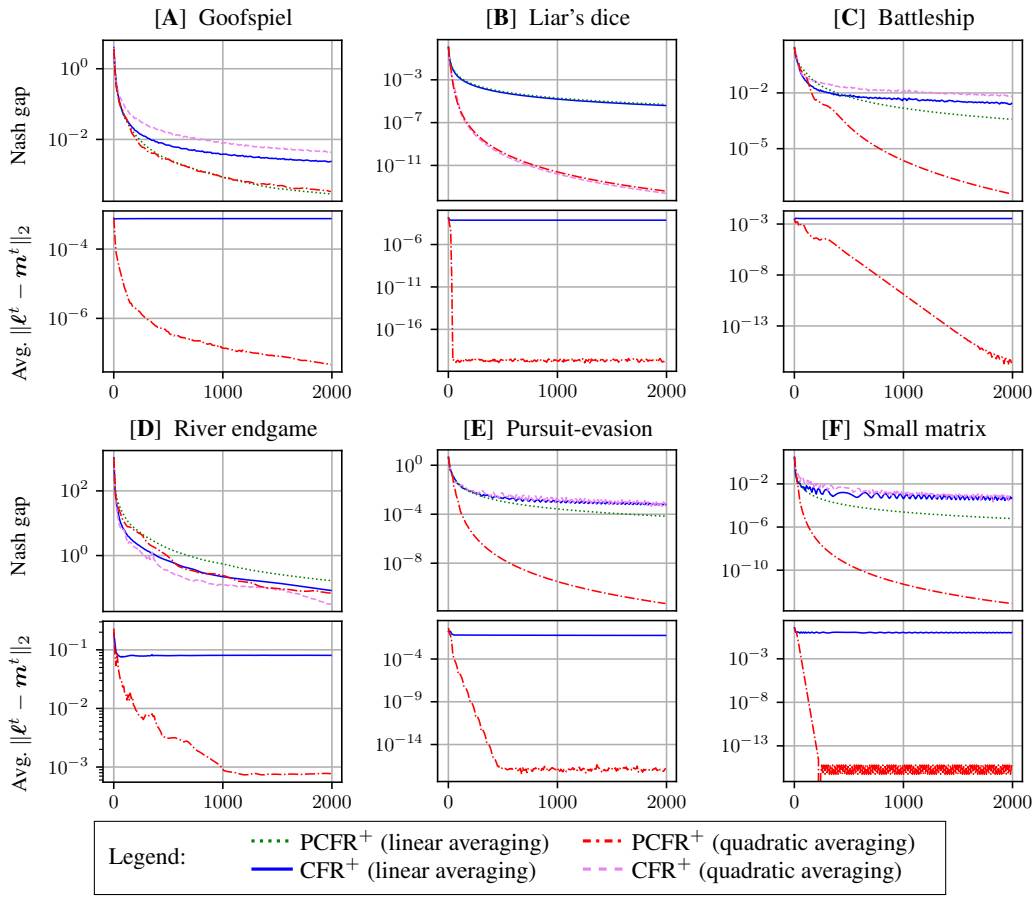
\centering%
    \addplotB{goof5}{Goofspiel}%
    \hspace{2mm}\addplotB*{ld_new}{Liar's dice}%
    \hspace{2mm}\addplotB*{bs23_4turns}{Battleship}\\
    \addplotB{river_endgame7}{River endgame}%
    \hspace*{2mm}\addplotB*{search6}{Pursuit-evasion}%
    \hspace*{2mm}\addplotB*{sm}{Small matrix}\\
    \centering\makelegendB{}
    \caption{Performance of \pcfrp{} and \cfrp{} with linear and quadratic averaging on EFGs. In all plots, the x axis is the number of iterations of each algorithm. For each game, the top plot shows that the Nash gap on the y axis (on a log scale), the bottom plot shows and the average prediction error (on a log scale).}\label{fig:comparison cfrp quadratic avg 1}
\end{figure}

\begin{figure}[H]\centering%
    \addplotB{kuhn}{Kuhn poker}%
    \hspace{2mm}\addplotB*{kuhn13}{Kuhn poker (13 ranks)}%
    \hspace{2mm}\addplotB*{search4}{Pursuit-evasion (4 turns)}\\
    \addplotB{search5}{Pursuit-evasion (5 turns)}%
    \hspace{2mm}\addplotB*{goof4}{Goofspiel 4}%
    \hspace{2mm}\addplotB*{goof4_li}{Goofspiel 4 (lim. info.)}\\
    \addplotB{goof5_li}{Goofspiel 5 (lim. info.)}%
    \hspace{2mm}\addplotB*{leduc3}{Leduc poker (3 ranks)}%
    \hspace{2mm}\addplotB*{leduc5}{Leduc poker (5 ranks)}\\
    \centering\makelegendB{}
    \caption{(continued) Performance of \pcfrp{} and \cfrp{} with linear and quadratic averaging on EFGs. In all plots, the x axis is the number of iterations of each algorithm. For each game, the top plot shows that the Nash gap on the y axis (on a log scale), the bottom plot shows and the average prediction error (on a log scale).}\label{fig:comparison cfrp quadratic avg 2}
\end{figure}

\begin{figure}[H]\centering%
    \addplotB{leduc9}{Leduc poker (9 ranks)}%
    \hspace{2mm}\addplotB*{leduc13}{Leduc poker (13 ranks)}%
    \hspace{2mm}\addplotB*{bs23_3turns}{Battleship (3 turns)}\\
    \caption{(continued) Performance of \pcfrp{} and \cfrp{} with linear and quadratic averaging on EFGs. In all plots, the x axis is the number of iterations of each algorithm. For each game, the top plot shows that the Nash gap on the y axis (on a log scale), the bottom plot shows and the average prediction error (on a log scale).}\label{fig:comparison cfrp quadratic avg 2}
\end{figure}

\subsection{Predictive Discounted CFR and Quadratic-Average Loss Prediction}

DCFR is the regret minimizer that results from applying the counterfactual regret minimization framework (\cref{app:efg cfr}) using the \emph{discounted regret matching} regret minimizer at each decision point. We experimentally evaluated a predictive-in-spirit\footnote{In fact, we do not have a proof that our variant is predictive in the formal sense described in the body of the paper. However, the variant we describe follows the natural pattern of predictive RM and predictive \rmp{}.} variant of discounted regret matching shown in \cref{algo:pdrm}.

    \begin{figure}[H]\centering
    \begin{minipage}[t]{.8\linewidth}\small
        \begin{algorithm}[H]\caption{Predictive discounted regret matching}\label{algo:pdrm}
            \DontPrintSemicolon
            $\vec{z}^0 \gets \vec{0} \in \bbR^n,\ \ \vec{x}^0 \gets \vec{1}/n \in \Delta^n$\;
            $\alpha \gets 1.5, \beta \gets 0$\;
            \Hline{}
            \Fn{\normalfont\textsc{NextStrategy}($\vec{m}^{t}$)}{
                \Comment{\color{commentcolor} Set $\vec{m}^t = \vec{0}$ for non-predictive version}\vspace{1mm}
                $\displaystyle\vec{\theta}^t \gets \frac{t^\alpha}{1+t^\alpha}[\vec{z}^{t-1}]^+ + \frac{t^\beta}{1+t^\beta}[\vec{z}^{t-1}]^- + \langle\vec{m}^t,\vec{x}^t\rangle \vec{1} - \vec{m}^t$\;
                \textbf{if} $\vec{\theta}^t \neq \vec{0}$ \textbf{return} $\vec{x}^t \gets \vec{\theta}^t \ /\ \|\vec{\theta}^t\|_1$\;
                \textbf{else} \hspace{0.715cm}\textbf{return} $\vec{x}^t \gets $ arbitrary point in $\Delta^{\!n}$\hspace*{-1cm}\;
            }
            \Hline{}
            \Fn{\normalfont\textsc{ObserveLoss}($\vec{\ell}^{t}$)}{
                $\displaystyle\vec{z}^t \gets \frac{t^\alpha}{1+t^\alpha}[\vec{z}^{t-1}]^+ + \frac{t^\beta}{1+t^\beta}[\vec{z}^{t-1}]^- + \langle\vec{\ell}^t,\vec{x}^t\rangle \vec{1} - \vec{\ell}^t$\;
            }
        \end{algorithm}
    \end{minipage}
    \end{figure}

To maintain symmetry with predictive CFR and predictive \cfrp{}, we coin \emph{predictive DCFR} the algorithm resulting from applying the  counterfactual regret minimization framework (\cref{app:efg cfr}) using the predictive discounted regret matching regret minimizer at each decision point of the game.

We also investigate the use of the quadratic average of past loss vectors,
\[
\vec{m}^t = \frac{6}{t(t-1)(2t-1)}\sum_{\tau=1}^{t-1} \tau^2 \vec{\ell}^\tau,
\]
as the prediction for the next loss $\vec{\ell}^t$. We call this loss prediction the ``quadratic-average loss prediction''.

We compare predictive DCFR (with and without quadratic-average loss prediction), \pcfrp{} (with and without quadratic-average loss prediction), \cfrp{}, and DCFR in \cref{fig:pdfcr quadratic pred 1,fig:pdfcr quadratic pred 2}.

\begin{figure}[H]
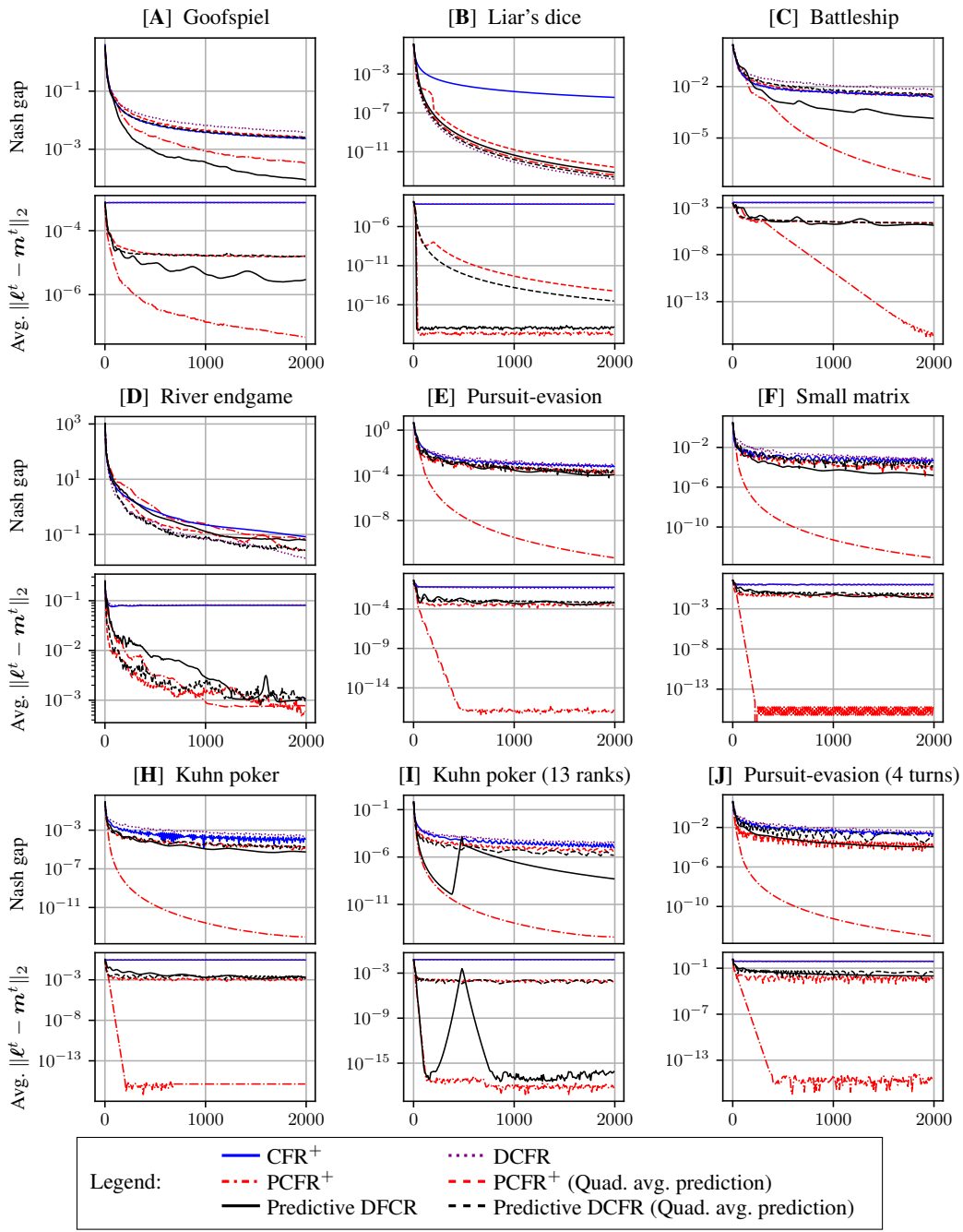
\centering%
    \addplotC{goof5}{Goofspiel}%
    \hspace{2mm}\addplotC*{ld_new}{Liar's dice}%
    \hspace{2mm}\addplotC*{bs23_4turns}{Battleship}\\
    \addplotC{river_endgame7}{River endgame}%
    \hspace{2mm}\addplotC*{search6}{Pursuit-evasion}%
    \hspace{2mm}\addplotC*{sm}{Small matrix}\\
    \addplotC{kuhn}{Kuhn poker}%
    \hspace{2mm}\addplotC*{kuhn13}{Kuhn poker (13 ranks)}%
    \hspace{2mm}\addplotC*{search4}{Pursuit-evasion (4 turns)}\\
    \centering\makelegendC{}
    \caption{Comparison between of discounted CFR and \cfrp{}, with and without quadratic-average loss prediction. In all plots, the x axis is the number of iterations of each algorithm. For each game, the top plot shows that the Nash gap on the y axis (on a log scale), the bottom plot shows and the average prediction error (on a log scale).}\label{fig:pdfcr quadratic pred 1}
\end{figure}

\begin{figure}[H]\centering%

    \addplotC{search5}{Pursuit-evasion (5 turns)}%
    \hspace{2mm}\addplotC*{goof4}{Goofspiel 4}%
    \hspace{2mm}\addplotC*{goof4_li}{Goofspiel 4 (lim. info.)}\\
    \addplotC{goof5_li}{Goofspiel 5 (lim. info.)}%
    \hspace{2mm}\addplotC*{leduc3}{Leduc poker (3 ranks)}%
    \hspace{2mm}\addplotC*{leduc5}{Leduc poker (5 ranks)}\\
    \addplotC{leduc9}{Leduc poker (9 ranks)}%
    \hspace{2mm}\addplotC*{leduc13}{Leduc poker (13 ranks)}%
    \hspace{2mm}\addplotC*{bs23_3turns}{Battleship (3 turns)}\\
    \centering\makelegendC{}
    \caption{(continued) Comparison between of discounted CFR and \cfrp{}, with and without quadratic-average loss prediction. In all plots, the x axis is the number of iterations of each algorithm. For each game, the top plot shows that the Nash gap on the y axis (on a log scale), the bottom plot shows and the average prediction error (on a log scale).}\label{fig:pdfcr quadratic pred 2}
\end{figure}

%% file: optimistic-approachability.bbl
\begin{thebibliography}{43}
\providecommand{\natexlab}[1]{#1}
\providecommand{\url}[1]{\texttt{#1}}
\providecommand{\urlprefix}{URL }
\expandafter\ifx\csname urlstyle\endcsname\relax
  \providecommand{\doi}[1]{doi:\discretionary{}{}{}#1}\else
  \providecommand{\doi}{doi:\discretionary{}{}{}\begingroup
  \urlstyle{rm}\Url}\fi

\bibitem[{Abernethy, Bartlett, and Hazan(2011)}]{Abernethy11:Blackwell}
Abernethy, J.; Bartlett, P.~L.; and Hazan, E. 2011.
\newblock {Blackwell} Approachability and No-Regret Learning are Equivalent.
\newblock In \emph{COLT}, 27--46.

\bibitem[{Blackwell(1954)}]{Blackwell54:Controlled}
Blackwell, D. 1954.
\newblock Controlled random walks.
\newblock In \emph{Proceedings of the international congress of
  mathematicians}, volume~3, 336--338.

\bibitem[{Blackwell(1956)}]{Blackwell56:analog}
Blackwell, D. 1956.
\newblock An analog of the minmax theorem for vector payoffs.
\newblock \emph{Pacific Journal of Mathematics} 6: 1--8.

\bibitem[{Bo{\v{s}}ansk{\`y} and {\v{C}}erm{\'a}k(2015)}]{Bosansky15:Sequence}
Bo{\v{s}}ansk{\`y}, B.; and {\v{C}}erm{\'a}k, J. 2015.
\newblock Sequence-form algorithm for computing {Stackelberg} equilibria in
  extensive-form games.
\newblock In \emph{Twenty-Ninth AAAI Conference on Artificial Intelligence}.

\bibitem[{Bo{\v{s}}ansk{\`y} et~al.(2014)Bo{\v{s}}ansk{\`y}, Kiekintveld,
  Lis{\'y}, and P{\v{e}}chou{\v{c}}ek}]{Bosansky14:Exact}
Bo{\v{s}}ansk{\`y}, B.; Kiekintveld, C.; Lis{\'y}, V.; and
  P{\v{e}}chou{\v{c}}ek, M. 2014.
\newblock An Exact Double-Oracle Algorithm for Zero-Sum Extensive-Form Games
  with Imperfect Information.
\newblock \emph{Journal of Artificial Intelligence Research} 829--866.

\bibitem[{Bowling et~al.(2015)Bowling, Burch, Johanson, and
  Tammelin}]{Bowling15:Heads}
Bowling, M.; Burch, N.; Johanson, M.; and Tammelin, O. 2015.
\newblock Heads-up Limit Hold'em Poker is Solved.
\newblock \emph{Science} 347(6218).

\bibitem[{Brown, Kroer, and Sandholm(2017)}]{Brown17:Dynamic}
Brown, N.; Kroer, C.; and Sandholm, T. 2017.
\newblock Dynamic Thresholding and Pruning for Regret Minimization.
\newblock In \emph{AAAI Conference on Artificial Intelligence (AAAI)}.

\bibitem[{Brown and Sandholm(2017)}]{Brown17:Superhuman}
Brown, N.; and Sandholm, T. 2017.
\newblock Superhuman {AI} for heads-up no-limit poker: {Libratus} beats top
  professionals.
\newblock \emph{Science} eaao1733.

\bibitem[{Brown and Sandholm(2019{\natexlab{a}})}]{Brown19:Solving}
Brown, N.; and Sandholm, T. 2019{\natexlab{a}}.
\newblock Solving imperfect-information games via discounted regret
  minimization.
\newblock In \emph{AAAI Conference on Artificial Intelligence (AAAI)}.

\bibitem[{Brown and Sandholm(2019{\natexlab{b}})}]{Brown19:Superhuman}
Brown, N.; and Sandholm, T. 2019{\natexlab{b}}.
\newblock Superhuman {AI} for multiplayer poker.
\newblock \emph{Science} 365(6456): 885--890.

\bibitem[{Burch(2018)}]{Burch18:Time}
Burch, N. 2018.
\newblock Time and space: Why imperfect information games are hard .

\bibitem[{Burch, Moravcik, and Schmid(2019)}]{Burch19:Revisiting}
Burch, N.; Moravcik, M.; and Schmid, M. 2019.
\newblock Revisiting {CFR}+ and alternating updates.
\newblock \emph{Journal of Artificial Intelligence Research} 64: 429--443.

\bibitem[{Chiang et~al.(2012)Chiang, Yang, Lee, Mahdavi, Lu, Jin, and
  Zhu}]{Chiang12:Online}
Chiang, C.-K.; Yang, T.; Lee, C.-J.; Mahdavi, M.; Lu, C.-J.; Jin, R.; and Zhu,
  S. 2012.
\newblock Online optimization with gradual variations.
\newblock In \emph{Conference on Learning Theory}, 6--1.

\bibitem[{Farina et~al.(2019{\natexlab{a}})Farina, Kroer, Brown, and
  Sandholm}]{Farina19:Stable}
Farina, G.; Kroer, C.; Brown, N.; and Sandholm, T. 2019{\natexlab{a}}.
\newblock Stable-Predictive Optimistic Counterfactual Regret Minimization.
\newblock In \emph{International Conference on Machine Learning (ICML)}.

\bibitem[{Farina, Kroer, and Sandholm(2019{\natexlab{a}})}]{Farina19:Online}
Farina, G.; Kroer, C.; and Sandholm, T. 2019{\natexlab{a}}.
\newblock Online Convex Optimization for Sequential Decision Processes and
  Extensive-Form Games.
\newblock In \emph{AAAI Conference on Artificial Intelligence (AAAI)}.

\bibitem[{Farina, Kroer, and
  Sandholm(2019{\natexlab{b}})}]{Farina19:Optimistic}
Farina, G.; Kroer, C.; and Sandholm, T. 2019{\natexlab{b}}.
\newblock Optimistic Regret Minimization for Extensive-Form Games via Dilated
  Distance-Generating Functions.
\newblock In \emph{Advances in Neural Information Processing Systems},
  5222--5232.

\bibitem[{Farina, Kroer, and Sandholm(2019{\natexlab{c}})}]{Farina19:Regret}
Farina, G.; Kroer, C.; and Sandholm, T. 2019{\natexlab{c}}.
\newblock Regret Circuits: Composability of Regret Minimizers.
\newblock In \emph{International Conference on Machine Learning}, 1863--1872.

\bibitem[{Farina, Kroer, and Sandholm(2020)}]{Farina20:Stochastic}
Farina, G.; Kroer, C.; and Sandholm, T. 2020.
\newblock Stochastic regret minimization in extensive-form games.
\newblock In \emph{International Conference on Machine Learning (ICML)}.

\bibitem[{Farina et~al.(2019{\natexlab{b}})Farina, Ling, Fang, and
  Sandholm}]{Farina19:Correlation}
Farina, G.; Ling, C.~K.; Fang, F.; and Sandholm, T. 2019{\natexlab{b}}.
\newblock Correlation in Extensive-Form Games: Saddle-Point Formulation and
  Benchmarks.
\newblock In \emph{Conference on Neural Information Processing Systems
  (NeurIPS)}.

\bibitem[{Foster(1999)}]{Foster99:Proof}
Foster, D.~P. 1999.
\newblock A proof of calibration via Blackwell's approachability theorem.
\newblock \emph{Games and Economic Behavior} 29(1-2): 73--78.

\bibitem[{Gao, Kroer, and Goldfarb(2021)}]{Gao19:Increasing}
Gao, Y.; Kroer, C.; and Goldfarb, D. 2021.
\newblock Increasing Iterate Averaging for Solving Saddle-Point Problems.
\newblock In \emph{AAAI Conference on Artificial Intelligence (AAAI)}.

\bibitem[{Gordon(2005)}]{Gordon05:NoRegret}
Gordon, G.~J. 2005.
\newblock No-regret algorithms for structured prediction problems.
\newblock Technical report, Carnegie-Mellon University, Computer Science
  Department, Pittsburgh PA USA.

\bibitem[{Gordon(2007)}]{Gordon07:NoRegret}
Gordon, G.~J. 2007.
\newblock No-regret algorithms for online convex programs.
\newblock In \emph{Advances in Neural Information Processing Systems},
  489--496.

\bibitem[{Hart and Mas-Colell(2000)}]{Hart00:Simple}
Hart, S.; and Mas-Colell, A. 2000.
\newblock A Simple Adaptive Procedure Leading to Correlated Equilibrium.
\newblock \emph{Econometrica} 68: 1127--1150.

\bibitem[{Hoda et~al.(2010)Hoda, Gilpin, Pe{\~n}a, and
  Sandholm}]{Hoda10:Smoothing}
Hoda, S.; Gilpin, A.; Pe{\~n}a, J.; and Sandholm, T. 2010.
\newblock Smoothing Techniques for Computing {N}ash Equilibria of Sequential
  Games.
\newblock \emph{Mathematics of Operations Research} 35(2).

\bibitem[{Kroer, Farina, and Sandholm(2018{\natexlab{a}})}]{Kroer18:Robust}
Kroer, C.; Farina, G.; and Sandholm, T. 2018{\natexlab{a}}.
\newblock Robust Stackelberg Equilibria in Extensive-Form Games and Extension
  to Limited Lookahead.
\newblock In \emph{AAAI Conference on Artificial Intelligence (AAAI)}.

\bibitem[{Kroer, Farina, and Sandholm(2018{\natexlab{b}})}]{Kroer18:Solving}
Kroer, C.; Farina, G.; and Sandholm, T. 2018{\natexlab{b}}.
\newblock Solving Large Sequential Games with the Excessive Gap Technique.
\newblock In \emph{Proceedings of the Annual Conference on Neural Information
  Processing Systems (NIPS)}.

\bibitem[{Kroer et~al.(2020)Kroer, Waugh, K{\i}l{\i}n{\c{c}}-Karzan, and
  Sandholm}]{Kroer20:Faster}
Kroer, C.; Waugh, K.; K{\i}l{\i}n{\c{c}}-Karzan, F.; and Sandholm, T. 2020.
\newblock Faster algorithms for extensive-form game solving via improved
  smoothing functions.
\newblock \emph{Mathematical Programming} .

\bibitem[{Kuhn(1950)}]{Kuhn50:Simplified}
Kuhn, H.~W. 1950.
\newblock A Simplified Two-Person Poker.
\newblock In Kuhn, H.~W.; and Tucker, A.~W., eds., \emph{Contributions to the
  Theory of Games}, volume~1 of \emph{Annals of Mathematics Studies, 24},
  97--103. Princeton, New Jersey: Princeton University Press.

\bibitem[{Lanctot et~al.(2009)Lanctot, Waugh, Zinkevich, and
  Bowling}]{Lanctot09:Monte}
Lanctot, M.; Waugh, K.; Zinkevich, M.; and Bowling, M. 2009.
\newblock {M}onte {C}arlo Sampling for Regret Minimization in Extensive Games.
\newblock In \emph{Proceedings of the Annual Conference on Neural Information
  Processing Systems (NIPS)}.

\bibitem[{Lis{\`y}, Lanctot, and Bowling(2015)}]{Lisy15:Online}
Lis{\`y}, V.; Lanctot, M.; and Bowling, M. 2015.
\newblock Online {M}onte {C}arlo counterfactual regret minimization for search
  in imperfect information games.
\newblock In \emph{Proceedings of the 2015 international conference on
  autonomous agents and multiagent systems}, 27--36.

\bibitem[{Morav{\v c}{\'\i}k et~al.(2017)Morav{\v c}{\'\i}k, Schmid, Burch,
  Lis{\'y}, Morrill, Bard, Davis, Waugh, Johanson, and
  Bowling}]{Moravvcik17:DeepStack}
Morav{\v c}{\'\i}k, M.; Schmid, M.; Burch, N.; Lis{\'y}, V.; Morrill, D.; Bard,
  N.; Davis, T.; Waugh, K.; Johanson, M.; and Bowling, M. 2017.
\newblock DeepStack: Expert-level artificial intelligence in heads-up no-limit
  poker.
\newblock \emph{Science} .

\bibitem[{Nesterov(2009)}]{Nesterov09:Primal}
Nesterov, Y. 2009.
\newblock Primal-dual subgradient methods for convex problems.
\newblock \emph{Mathematical programming} 120(1): 221--259.

\bibitem[{Rakhlin and Sridharan(2013{\natexlab{a}})}]{Rakhlin13:Online}
Rakhlin, A.; and Sridharan, K. 2013{\natexlab{a}}.
\newblock Online Learning with Predictable Sequences.
\newblock In \emph{Conference on Learning Theory}, 993--1019.

\bibitem[{Rakhlin and Sridharan(2013{\natexlab{b}})}]{Rakhlin13:Optimization}
Rakhlin, S.; and Sridharan, K. 2013{\natexlab{b}}.
\newblock Optimization, learning, and games with predictable sequences.
\newblock In \emph{Advances in Neural Information Processing Systems},
  3066--3074.

\bibitem[{Ross(1971)}]{Ross71:Goofspiel}
Ross, S.~M. 1971.
\newblock Goofspiel—the game of pure strategy.
\newblock \emph{Journal of Applied Probability} 8(3): 621--625.

\bibitem[{Shalev-Shwartz and Singer(2007)}]{Schwartz07:Primal}
Shalev-Shwartz, S.; and Singer, Y. 2007.
\newblock A primal-dual perspective of online learning algorithms.
\newblock \emph{Machine Learning} 69(2-3): 115--142.

\bibitem[{Southey et~al.(2005)Southey, Bowling, Larson, Piccione, Burch,
  Billings, and Rayner}]{Southey05:Bayes}
Southey, F.; Bowling, M.; Larson, B.; Piccione, C.; Burch, N.; Billings, D.;
  and Rayner, C. 2005.
\newblock {Bayes}' Bluff: Opponent Modelling in Poker.
\newblock In \emph{Proceedings of the 21st Annual Conference on Uncertainty in
  Artificial Intelligence (UAI)}.

\bibitem[{Syrgkanis et~al.(2015)Syrgkanis, Agarwal, Luo, and
  Schapire}]{Syrgkanis15:Fast}
Syrgkanis, V.; Agarwal, A.; Luo, H.; and Schapire, R.~E. 2015.
\newblock Fast convergence of regularized learning in games.
\newblock In \emph{Advances in Neural Information Processing Systems},
  2989--2997.

\bibitem[{Tammelin(2014)}]{Tammelin14:Solving}
Tammelin, O. 2014.
\newblock Solving large imperfect information games using CFR+.
\newblock \emph{arXiv preprint arXiv:1407.5042} .

\bibitem[{{von Stengel}(1996)}]{Stengel96:Efficient}
{von Stengel}, B. 1996.
\newblock Efficient Computation of Behavior Strategies.
\newblock \emph{Games and Economic Behavior} 14(2): 220--246.

\bibitem[{Waugh and Bagnell(2015)}]{Waugh15:Unified}
Waugh, K.; and Bagnell, D. 2015.
\newblock A Unified View of Large-scale Zero-sum Equilibrium Computation.
\newblock In \emph{Computer Poker and Imperfect Information Workshop at the
  AAAI Conference on Artificial Intelligence (AAAI)}.

\bibitem[{Zinkevich et~al.(2007)Zinkevich, Bowling, Johanson, and
  Piccione}]{Zinkevich07:Regret}
Zinkevich, M.; Bowling, M.; Johanson, M.; and Piccione, C. 2007.
\newblock Regret Minimization in Games with Incomplete Information.
\newblock In \emph{Proceedings of the Annual Conference on Neural Information
  Processing Systems (NIPS)}.

\end{thebibliography}
